\documentclass[12pt]{article}

\usepackage[utf8]{inputenc}
\usepackage[T1]{fontenc}
\usepackage[english]{babel}

\usepackage[a4paper, total={6.3in, 8.9in}]{geometry}

\usepackage{tocloft}
\usepackage{mathptmx} 

\usepackage{verbatim}
\usepackage{enumitem}
\usepackage{array,multirow}
\usepackage{amssymb,amsfonts,amsmath,amsthm}
\usepackage{aliascnt} 
\usepackage{url}
\usepackage{tikz}
\usepackage{bbm}
\usepackage{accents}
\usepackage{thmtools}

\usepackage{placeins}

\usetikzlibrary{arrows.meta}
\usetikzlibrary{calc}

\usepackage{xcolor}
\usepackage{ctable}
\usetikzlibrary{matrix,arrows}
\usetikzlibrary{decorations.markings}

\definecolor{citec}{HTML}{324FDA} 
\definecolor{linkc}{HTML}{A0111A}
\definecolor{urlc}{HTML}{b55c87}
\usepackage[colorlinks,
    citecolor=citec,
    linkcolor=linkc,
    urlcolor=urlc,
    bookmarks = true,
    breaklinks,
]{hyperref}
\usepackage[nameinlink, noabbrev, capitalize]{cleveref}

\theoremstyle{plain}
\newtheorem{theorem}{Theorem}[section]

\newtheorem{lemma}[theorem]{Lemma}
\newtheorem{corollary}[theorem]{Corollary}
\newtheorem{proposition}[theorem]{Proposition}

\newtheorem{definition}[theorem]{Definition}
\newtheorem{fact}[theorem]{Fact}
\newtheorem{remark}[theorem]{Remark}

\hypersetup{hypertexnames=false}

\AtBeginDocument{%
}

\numberwithin{equation}{section}

\newcommand{\eps}{\varepsilon}
\newcommand{\Tr}{\operatorname{Tr}}

\def\poly{{\mathrm{poly}}}

\DeclareMathOperator{\im}{im}

\renewcommand{\Pr}{\mathop{\bf Pr\/}}

\newcommand{\FF}{{\mathbb{F}}}

\newcommand{\ZZ}{{\mathbb{Z}}}

\begin{document}
\pagenumbering{gobble} 

\title{Algebraic-Geometric Parvaresh--Vardy Subspace Designs and Rank Condensers}

\date{}
        \author{
  Gil Cohen\thanks{Tel Aviv University. \texttt{gil@tauex.tau.ac.il}. Supported by ERC starting grant 949499 and by the Israel Science Foundation grant 2989/24.  	
  }\and
  Dean Doron\thanks{Ben Gurion University. \texttt{deand@bgu.ac.il}. Supported in part by NSF-BSF grant 2022644.}
  \and
  Noam Goldgraber\thanks{Ben Gurion University and Tel Aviv University. \texttt{goldgrab@post.bgu.ac.il}. Supported by NSF-BSF grant 2022644.}
}

\maketitle

\begin{abstract}
A subspace design is a collection of subspaces $H_1,\ldots,H_n$ of $\mathbb{F}_q^k$ with the property that no low-dimensional subspace $W$ intersects the collection ``too much’’. Subspace designs and related objects in linear-algebraic pseudorandomness have found a broad range of applications, ranging from list decoding and recovery, to derandomizing algorithms. 

We construct explicit strong subspace designs over \emph{every} finite field. In the extremal case where the co-dimension $t$ of each $H_i$ is equal to the dimension of $W$, for every constant field size our construction attains $n=\Omega(k)$ and matches the probabilistic intersection bound up to a constant factor. All previous constructions required the field size to grow with $t$ (or $k$).
Our subspace designs also imply new constructions of rank condensers over arbitrary finite fields. This result is the first to achieve an optimal dependence on $k$ while maintaining both a constant output entropy rate and a constant field size.
As an application, we construct lossless rank extractors for linear sources of rank $r$, for all $r < q$, with parameters matching those of Guo, Raj, Shangguan and Zhang (FOCS '26), thereby generalizing their result to prime fields and smaller field sizes.

Our construction is based on an algebraic-geometric version of the Parvaresh--Vardy codes (Parvaresh--Vardy FOCS '05, Guruswami ECCC '05), extending the framework underlying the condensers of Guruswami, Umans and Vadhan (JACM '09). We view our construction as a linear-algebraic analysis -- tailored to affine sources -- of the GUV construction, generalized to functions over algebraic curves.
More specifically, inspired by Ta-Shma and Umans (CCC '12) we develop a two-level evaluation scheme, where we first evaluate a function on a curve at extension-field points, and then evaluate a corresponding affine-linear polynomial to obtain outputs over the base field.
\end{abstract}

\newpage
\tableofcontents
\newpage
\pagenumbering{arabic}  

\section{Introduction}

Algebraic pseudorandomness studies algebraic analogues of traditional, combinatorial, pseudorandom objects, wherein the families of tests we wish to fool are characterized by \emph{linear subspaces} rather than arbitrary sets. Often, we seek to efficiently construct a collection of linear subspaces that enjoy random-like properties with respect to an arbitrary linear subspace. Such examples include subspace designs, rank extractors and condensers, and dimension expanders
(see \cite{GR08a, FS12, GK16, FG15, CI17, GXY18} for a sample of earlier works). 
Beyond being natural and interesting in their own right, the linear-algebraic pseudorandom objects have found many applications, the earlier of which include, e.g., list decoding \cite{GX12}, randomness extractors \cite{GR08a}, and polynomial identity testing \cite{KS11, FS12}.

Very recently, research in linear-algebraic pseudorandomness has regained renewed momentum,
with new constructions \cite{GRSZ26, GGH26} and exciting applications. These applications include,
most prominently, subspace design codes and new results in list decoding and list recovery
\cite{CZ25, BCDZ26, GGH26, GG26}. Other applications include strong blocking sets \cite{GRSZ26}, tensor rank \cite{dvir2025}, and the breakthrough result of bipartite matching in $\mathbf{NC}$ \cite{CGGRT26} (see also \cite{KS26}).

\paragraph{Subspace Designs.} The object this work focuses on 
 is \emph{subspace designs}, first defined by Guruswami and Xing \cite{GX12}.

\begin{definition}[Subspace design]
Let $q$ be a prime power, let $1\leq r\leq t\leq k$, let $n\geq1$ be an integer and $A\geq0$, and let
$\mathcal{H}=(H_i)_{i\in[n]}$ be a collection of subspaces of
$\FF_q^k$, each of codimension at most $t$. We say that $\mathcal{H}$ is an
\emph{$(r,A)$-weak subspace design} if every $r$-dimensional subspace
$W\leq\FF_q^k$ satisfies
\[
    \bigl|\{i\in[n]:H_i\cap W\neq\{0\}\}\bigr|\leq A.
\]
We say that $\mathcal{H}$ is an \emph{$(r,A)$-strong subspace design} if for
every such $W$ we have
\[
    \sum_{i=1}^n\dim(H_i\cap W)\leq A.
\]
In either case, we say that the design has codimension $t$.
\end{definition}

Naturally, our goal is to minimize $A$ and maximize $n$, while keeping $t$ as close to $r$ as possible.
In this paper, our primary focus is on strong subspace designs, which we simply refer to as \emph{subspace designs}.

The probabilistic method guarantees the existence of a subspace design with $n = q^{\Omega(t)}$ 
and $A = O(rk/t)$ when $t \geq 2r$\footnote{More precisely, this holds for any $t \geq \alpha r$ where $\alpha > 1$.} (see \cite[Lemma 4]{GK16}). 
In the regime $r = t$,\footnote{This is the extremal case, as the sum of the dimensions of a subspace in the collection and the test subspace exactly equals the dimension of the ambient space. Any increase in either dimension trivially forces a non-zero intersection; thus, this is the regime where one expects the largest bound $A$.} the behavior is quite different. There, it yields $n = \Omega(r(k-r)\,\log^{1/2}(q))$ and $A = (1+o_q(1))r(k-r)$ (see \cite[Theorem~A.1]{GRSZ26}).

Our goal is to provide \emph{explicit} constructions,
meaning we seek an algorithm that outputs such an object in time $\text{poly}(k, q, n, A)$.
Similarly, throughout this paper, by \emph{explicit} we mean that the object can be constructed by a deterministic algorithm running in time polynomial in the relevant parameters. We allow the running time to be $\text{poly}(q)$ rather than $\text{poly}(\log q)$ because whenever $q > \text{poly}(k)$, there are already strictly better constructions \cite{GK16} that achieve polylogarithmic time. We therefore disregard this technicality and simply bound the running time by $\text{poly}(q)$.

When $t \geq 2r$,
the first explicit constructions with essentially optimal intersection bounds,
due to Guruswami and Kopparty \cite{GK16}, require $q>k$. Later, Guruswami, Xing, and
Yuan \cite{GXY18} generalized the \cite{GK16} 
construction to function fields, which provided a construction over every finite field, but with the
weaker bound
\[
    A=O\left(\frac{rk\,\log_q k}{t-r+1}\right).
\]
When $t = O(r)$, the size of this collection is $n = \frac{k}{t}q^{O(1)}$, which again requires $q$ to grow with $k$ in order to provide a non-trivial result.
Recently, Goyal, Guruswami and Hsieh \cite{GGH26} used the expander-based Alon-Edmonds-Luby (AEL)
framework to construct subspace designs over any finite field with $A = O(rk/t)$, but with $t = \text{poly}(r)\cdot q^{r^2}$.

In this paper, our focus is on the regime $r = t$.
Guo, Raj, Shangguan, and Zhang \cite{GRSZ26} recently
constructed explicit subspace designs with $A = O(rk)$ and $n \geq q^{1/4}A$, for all non-prime field sizes
$q$ such that $q \geq \text{poly}(r)$. For prime fields in this regime, and for smaller field sizes, they obtain
subspace designs with larger bounds on $A$; That is, either multiplied by a superpolynomial factor in $r$, or by an exponential factor in $r$, respectively.
Hence, their construction achieves the optimal value of $A$ up to a constant, whenever $q \geq \text{poly}(r)$.
This is the first construction that achieves near-optimal parameters and field size independent of the ambient-space dimension $k$,
but rather depends only on the dimension of the test subspace $r$.

The main contribution of this work is a new, explicit construction of strong subspace designs.
\begin{theorem}[\cref{cor:all-field-strong-subspace-designs}, simplified]\label{thm:intro-subspace-designs}
    Let $q$ be a prime power, let $0 < \gamma < 1$, and let $k\geq2$. For every
    $1\leq r\leq k$,
    there exists an explicit $(r,A)$ subspace design $\mathcal{H}$ of $n$ subspaces of
    $\FF_q^k$, each of codimension at most $r$, satisfying
    \[
        A=O\!\left(\frac{q^3rk}{\gamma^5}\right),
        \qquad
        \frac{A}{n}\le \left(\frac1q+\gamma\right)r.
    \]
    Moreover, every nonzero subspace $W\leq\FF_q^k$ of dimension at most
    $r$ satisfies
    \[
        \mathbb{E}_{H\in\mathcal{H}}
        \dim(W\cap H)
        <\left(\frac{1}{q}+\gamma\right)\dim W.
    \]
\end{theorem}
Crucially, our construction applies to arbitrary field sizes while simultaneously achieving $t = r$.
Moreover, it achieves the optimal value of $A$ up to a $\text{poly}(q)$ factor, which is simply
a constant when $q$ itself is constant. For fields of square size, we further improve this bound
to $A = 6rk$ (see \cref{thm:square-field-subspace-design}).

We also consider related linear-algebraic pseudorandom objects, whose constructions follow primarily from our results on subspace designs.

\paragraph{Lossy rank condensers.} We first consider \emph{lossy rank condensers}, originally introduced by Forbes, Saptharishi and Shpilka in \cite{FSS14}
for applications to read-once algebraic branching programs (though they specifically considered the \emph{lossless} variant; see \cref{def:lossless-rank-extractors}).
Later, Forbes and Guruswami \cite{FG15}
constructed lossy rank condensers via a reduction from the subspace designs of \cite{GK16},
and applied them to construct dimension expanders. We begin with the definition.

\begin{definition}[Lossy rank condenser]\label{def:lossy-rank-condenser}
Let $q$ be a prime power, let $1\leq r\leq t\leq k$, let $n\geq1$ be an integer, let
$\varepsilon,\delta\in[0,1]$, and let
$\mathcal{E}=(E_i)_{i\in[n]}$ be a collection of matrices in
$\FF_q^{t\times k}$. We say that $\mathcal{E}$ is an
\emph{$(r,\varepsilon,\delta)$-lossy rank condenser} if every rank-$r$
matrix $M\in\FF_q^{k\times r}$ satisfies
\[
    \bigl|\{i\in[n]:
    \operatorname{rank}(E_iM)<(1-\varepsilon)r\}\bigr|
    \leq\delta n.
\]
\end{definition}
In \cite{FG15}, lossy rank condensers were defined as the
special case $\delta=1-1/n$ of the above definition. In this case, for every
rank-$r$ matrix $M$, at least one matrix $E_i$ satisfies
$\operatorname{rank}(E_iM)\geq(1-\varepsilon)r$. We find it more natural to refer to this object as a \textit{lossy rank disperser}.

Our subspace designs yield lossy rank condensers with output dimension $t = r$.
\begin{theorem}[\cref{cor:all-field-lossy-rank-condensers}]
Let $q$ be a prime power, and let $\varepsilon,\delta\in(0,1]$ satisfy
$q>1/(\varepsilon\delta)$. For every
integers $1\leq r\leq k$,
there exists an explicit $(r,\varepsilon,\delta)$-lossy rank condenser
$\mathcal{E}=(E_i)_{i\in[n]}\subseteq\FF_q^{r\times k}$ of size
\[
    n=O\left(\frac{q^3k}{(\varepsilon\delta-1/q)^5}\right).
\]
\end{theorem}
The primary feature of this construction is that it works over any constant field size
$q$ while maintaining a constant output entropy rate -- by choosing a constant $\varepsilon < 1$. Previously, all constructions required
either a growing field size (e.g., \cite{FG15,GRSZ26}) or a vanishing output entropy rate
(e.g., \cite{GXY18,GGH26}) -- that is, the rank of their promised output subspace was sublinear in the dimension of the output space\footnote{An exception is \cite{GUV09}, which allows any field and achieves constant entropy, but with significantly larger $n$.}. We refer the reader to \cref{subsec:affine-source-condensers},
where we also spell out the implications of our result (and previous results) to linear seeded
condensers for affine sources, which might be useful from a pseudorandomness perspective.

\paragraph{Lossless rank extractors.} A \emph{lossless rank extractor} is simply an $(r,0, \delta)$ lossy rank condenser in which $r=t$, and therefore retains all the entropy.
\begin{definition}[Lossless rank extractor]\label{def:lossless-rank-extractors}
Let $q$ be a prime power, let $1\leq r\leq k$, let $n\geq1$, and let
$\mathcal{E}=(E_i)_{i\in[n]}$ be a collection of matrices in
$\FF_q^{r\times k}$. For $\delta>0$, we say that $\mathcal{E}$ is an
\emph{$(r,\delta)$-lossless rank extractor} if every rank-$r$ matrix
$M\in\FF_q^{k\times r}$ satisfies
\[
    \bigl|\{i\in[n]:\operatorname{rank}(E_iM)<r\}\bigr|\leq\delta n.
\]
The collection $\mathcal{E}$ is called an \emph{$r$-lossless rank disperser} if for every rank-$r$ matrix
$M\in\FF_q^{k\times r}$ there exists $i\in[n]$ such that $\operatorname{rank}(E_iM)=r$.
\end{definition}
Observe that to obtain $r$-lossless rank dispersers over $\FF_q$, it suffices to construct a subspace design
with $t=r$ such that $A < n$. Any subset of size $A+1$ from this design then yields an $r$-lossless rank disperser.
This observation was used in both \cite{FG15} and \cite{GRSZ26}. The former, which relied on the subspace designs of
\cite{GK16}, required $q > r(k-r)$. The latter provided the first construction over a field size independent of $k$,
requiring only $q \geq \text{poly}(r)$. Furthermore, for non-prime fields they
achieved $A = O(rk)$ for infinitely many $k \geq r$ (or $A = O(qrk)$ for all $k \geq r$), whereas over large prime fields they incurred the weaker bound
$A = O(r^{O(\log r)}rk)$. By contrast, our construction improves the field size requirement to $q > r/\delta$ and applies to both prime and non-prime fields, at the cost of increasing $A$ by a $\text{poly}(r,q)$ multiplicative factor.
\begin{theorem}[\cref{cor:all-field-lossless-rank-extractors}]
Let $1\leq r\leq k$ be integers, and let $0<\delta<1$. For every prime power
$q>r/\delta$,
there exists an explicit $(r,\delta)$-lossless rank extractor
$\mathcal{E}=(E_i)_{i\in[n]}\subseteq\FF_q^{r\times k}$ of size
\[
    n=O\left(\frac{q^3r^5k}{(\delta-r/q)^5}\right).
\]
In particular, for every prime power $q>r$, there exists an explicit
$r$-lossless rank disperser over $\FF_q$ of size $O(r^5kq^8)$.
\end{theorem}

\paragraph{Strong blocking sets.}
Introduced by \cite{RCH56}, blocking sets are classical objects in finite geometry, that characterize the minimal structure necessary to guarantee large intersection with all subspaces of a given codimension. More formally, for $1\leq s\leq k-1$, a strong $s$-blocking set in
$\operatorname{PG}(k-1,q)$ satisfies that its intersection with an arbitrary codimension-$s$ projective
subspace, spans this subspace. The challenge is to explicitly construct strong blocking sets of minimal size. As shown in \cite{FS14lines, FS16lines},
lossless rank dispersers imply strong blocking sets. Applying this reduction to our $(s+1)$-lossless rank dispersers gives,
for every prime power $q>s+1$ and every $k\geq s+1$, an explicit strong
$s$-blocking set of size $O(s^5kq^{s+8})$; see
\cref{cor:prime-field-strong-blocking-sets}.
Let $c_0$ be the absolute
constant in \cite[Corollary~6.2]{GRSZ26}. When $q< (2s)^{c_0}$, our bound $O\left (kq^s s^{O(1)} \right)$ improves all applicable previous bounds, by an exponential factor.
For a prime
$(2s)^{c_0}\leq q \leq \text{poly}(s)$, it improves upon previous results by a quasipolynomial factor in $s$. We give a detailed discussion, together with the formal definitions and statements, in
\cref{subsec:strong-blocking-sets}.

\subsection{Proof technique}
As in the work of Guruswami, Umans, and Vadhan \cite{GUV09}, our subspace designs are based on Parvaresh--Vardy codes. Conceptually, our basic construction can be viewed as (a function field generalization of) the GUV construction, equipped with a linear-algebraic analysis that exploits the affine structure of the sources.
In this subsection, we adopt the dual perspective of rank condensers to draw a more direct comparison with their construction.

The GUV construction evaluates low degree polynomials over a finite field, and the seed of the condenser is parametrized by field elements. In the context of rank condensers, this means that each map corresponds to a unique element of the underlying field -- and therefore the field size must grow together with other parameters.

To circumvent this, we instead use functions over algebraic curves, with a growing number of points.
Combined with our linear-algebraic analysis, this approach yields lossy rank condensers over fields of square size, larger than $49$.
To extend this result to arbitrary finite fields, we apply a second evaluation step, following a technique from \cite{TSU12}.
We now describe these ideas in more detail.

\paragraph{The GUV construction.}
GUV construct their condenser by evaluating the correlated polynomials of
Parvaresh--Vardy codes \cite{PV05}.
Let $k\geq2$ and $1\leq r\leq k$, and identify $\FF_q^k$ with the space
$V$ of polynomials in $\FF_q[X]$ of degree less than $k$. Fix an
irreducible polynomial $Q\in\FF_q[X]$ of degree $k$ and an integer
$h\geq2$. For $f\in V$ and $0\leq i\leq r-1$, define
\[
    \tau_i(f)=f^{h^i}\bmod Q,
\]
considered as a polynomial of degree less than $k$. In particular, $\tau_0(f)=f$.
For an element $a\in\FF_q$, the GUV construction is
\[
    E_a(f)=\bigl(\tau_0(f)(a),\ldots,\tau_{r-1}(f)(a)\bigr),
    \qquad
    \operatorname{Cond}(f,a)
    =\bigl(a,E_a(f)\bigr).
\]

The GUV analysis follows the list-decoding argument for PV codes. Given an
output set $T\subseteq\FF_q^{r+1}$, we seek to bound, in terms of $|T|$,
the size of the list
\[
    \operatorname{LIST}(T)
    =\{f\in V:\operatorname{Cond}(f,a)\in T
      \text{ for every }a\in\FF_q\}.
\]
We first interpolate a nonzero polynomial
$R\in\FF_q[X,Y_0,\ldots,Y_{r-1}]$ vanishing on $T$. For every
$f\in\operatorname{LIST}(T)$, we have
\[
    R\bigl(a,\tau_0(f)(a),\ldots,\tau_{r-1}(f)(a)\bigr)=0
    \qquad\text{for every }a\in\FF_q.
\]
Provided that $T$ is sufficiently small, we can choose the degree bounds
on $R$ so that the univariate polynomial
$R(X,\tau_0(f),\ldots,\tau_{r-1}(f))$ has degree less than $q$.
Since it vanishes on all $q$ field elements, it is identically zero.

We now view each $f\in V$ as an element of the extension field
$\FF_q[X]/(Q)$. The relations $\tau_i(f)\equiv f^{h^i}\!\!\!\mod Q$
imply that every $f\in\operatorname{LIST}(T)$ is a root of
\[
    R^*(Z)=R(X,Z,Z^h,\ldots,Z^{h^{r-1}})\bmod Q
    \in\bigl(\FF_q[X]/(Q)\bigr)[Z].
\]
The monomials in $Y_0,\ldots,Y_{r-1}$ are chosen to become distinct powers
of $Z$, and $R$ can be chosen to remain nonzero modulo $Q$.
Thus, $R^*$ is nonzero, which yields
$|\operatorname{LIST}(T)|\leq\deg R^*$.

\paragraph{A linear-algebraic analysis.}
Linear condensers for affine sources can equivalently be viewed as lossy rank condensers (as we explain in \cref{subsec:affine-source-condensers}).
Henceforth, we proceed using the latter formulation.
Similarly to \cite[Corollary 2.23]{Che10}, we instantiate the same construction with $h=q$. With this choice, the maps $\tau_i$ and $E_a$ become $\FF_q$-linear, and our lossy rank condenser is simply the collection $\{E_a\}_{a \in \FF_q}$. We next present our analysis.

Fix an $s$-dimensional subspace $W\leq V$, where
$1\leq s\leq r$. The rank loss of ${E_a}$ at $W$ is $s-\dim E_a(W)$.
Our goal is to bound the total rank loss over all evaluation points $a$.
We obtain this bound from the determinant of a matrix associated with a
basis of $W$.

Let $\mathcal{B}=(f_1,\ldots,f_s)$ be an $\FF_q$-basis of $W$. We define
the \emph{PV-Moore matrix} by
\begin{equation}\label{eq:PV-Moore}
    M_{\mathcal{B}}
    =\bigl(\tau_i(f_j)\bigr)_%
      {0\leq i\leq s-1,\,1\leq j\leq s}.
\end{equation}
Reducing $M_{\mathcal{B}}$ modulo $Q$ gives the Moore matrix of the
residue classes of $f_1,\ldots,f_s$, in $\FF_q[X]/(Q) \cong \FF_{q^k}$. Since reduction modulo $Q$ is
injective on $V$, these classes are linearly independent over $\FF_q$.
The Moore-matrix criterion therefore implies that
$\Delta_{\mathcal{B}}:=\det M_{\mathcal{B}}$ is nonzero.

We next relate the rank loss of $E_a$ to the vanishing of
$\Delta_{\mathcal{B}}$ there. The matrix $M_{\mathcal{B}}(a)$ represents
the first $s$ coordinates of $E_a|_W$, so its kernel has dimension at least
$s-\dim E_a(W)$. On the other hand, the dimension of the kernel is at most the multiplicity of $a$ as a
root of $\Delta_{\mathcal{B}}$: by multiplying by an invertible matrix with $\FF_q$-coefficients, we can obtain $\dim \text{ker}(E_a|_W)$ columns of $M_{\mathcal{B}}$ divisible by $X-a$.
Since each entry has degree less than $k$, the determinant has degree at
most $s(k-1)$. Summing over all evaluation points gives
\[
    \sum_{a\in\FF_q}\bigl(s-\dim E_a(W)\bigr)
    \leq s(k-1).
\]
This bound is non-trivial only whenever $q \geq k$; consequently, the field size must grow.

\paragraph{Passing to function fields.}
To keep $q$ fixed while increasing the number of evaluation points, we
pass to function fields. This extension of the Parvaresh--Vardy framework was previously studied by Guruswami \cite{Gur05} in the context of list decoding. Consequently, his analysis relies on polynomial interpolation, and is similar to the later GUV analysis presented above. The idea is to replace low-degree polynomials,
which are evaluated over elements in $\FF_q$, with a space of ''low-degree`` functions over an algebraic curve,
which can be evaluated over the curve points. 
In our rank condensers setting, by taking curves with many evaluation points relative to their ''complexity`` -- measured by their \emph{genus} -- we are able to increase the number of evaluation points, and hence the number of linear transformations, without increasing the field size.
We provide the formal definitions in \cref{subsec:function-field-prelim}.

Let $1 \leq r \leq k$, and let $F/\FF_q$ be a function field of genus $g \leq k$ with $N$ rational places (i.e., curve points), plus an additional rational place $P_\infty$.
We identify $\FF_q^k$ with $V=\mathcal{L}(dP_\infty)$,
where $d=k+g-1$ -- that is, the space of functions of ''degree`` at most $d$. By the Riemann--Roch theorem, it has dimension $k$.
Analogous to the choice of an irreducible polynomial of degree $k$, choose a \emph{place} $Q$ of degree $m = d+1$.\footnote{More precisely, $m$ is chosen to be a prime number between $d$ and $2d$, which allows us to find such a place explicitly (see \cref{prop:explicit-prime-degree-place}). For the sake of simplicity, we omit this technicality here.}
Since $m>d$, the residue map $\pi_Q$ (i.e., reducing a function ''$\text{mod}\ Q$``) is injective on $V$, so
each function $f \in V$ is uniquely represented by an element of
$\FF_{q^m}$.

To define the correlated functions $\tau_i(f)$, we lift powers of
$\pi_Q(f)$ back to functions in $F$. These powers need not correspond
to functions in $V$, so we allow the lifts to have slightly larger
``degree''. Following \cite{Gur05}, we use the Riemann--Roch theorem
to choose these lifts linearly over $\FF_q$, with ``degree'' at most
$\widetilde d=m+2g-1<6k$. Taking $\tau_0(f)=f$, the resulting maps satisfy
\[
    \tau_i(f)=f^{q^i} \mod Q
    \qquad (0\leq i\leq r-1).
\]
Thus, the functions $\tau_i(f)$ satisfy the same relations
modulo $Q$ as in the polynomial construction.

Let $\mathcal{P}$ denote the $N$ rational evaluation places. Evaluating
the functions $\tau_i(f)$ at each $P\in\mathcal{P}$ gives maps
$E_P\colon V\to\FF_q^r$ as before. The same determinant argument now
applies. For an $s$-dimensional subspace $W\leq V$, where $1\leq s\leq r$,
the PV-Moore matrix of a basis $\mathcal{B}$ of $W$, defined similarly to \cref{eq:PV-Moore}, reduces to a
nonsingular Moore matrix modulo $Q$, so its determinant
$\Delta_{\mathcal{B}}$ is nonzero. This determinant has ``degree'' at most
$d+(s-1)\widetilde d$, and hence at most that many rational zeros,
counted with multiplicity. As before, the rank loss of $E_P$ is bounded
by the multiplicity of $P$ as a zero of $\Delta_{\mathcal{B}}$. Thus,
\[
    \sum_{P\in\mathcal{P}}\bigl(s-\dim E_P(W)\bigr)
    \leq d+(s-1)\widetilde d<6sk,
\]
see \cref{prop:strong-subspace-design}.

To make the construction explicit, we must find $Q$ efficiently; we accomplish this for the Garcia--Stichtenoth tower in \cref{prop:explicit-prime-degree-place}.

This tower has many rational
places relative to its genus. Evaluating it in the above construction gives, for every square $q$ and infinitely many $k$, an explicit
collection of $n=N$ linear maps satisfying
\[
    n\geq(\sqrt q-1)k,
    \qquad
    \frac{1}{n}\sum_{P\in\mathcal{P}}\bigl(s-\dim E_P(W)\bigr)
    <\frac{6s}{\sqrt q-1};
\]
see \cref{thm:square-field-subspace-design}. These maps construct an $(r, 6rk)$ subspace design of size $n$, or dually, an $(r, \varepsilon, \delta)$-lossy rank condenser provided that $6/(\sqrt{q}-1) \leq \varepsilon\delta$.

Two restrictions remain. The tower used in this instantiation requires
$q$ to be a square, and the estimate above is non-trivial only whenever $q > 49$. We next modify the
construction to overcome both restrictions.

\paragraph{Two evaluations over arbitrary finite fields.}
To relax the requirement $6/(\sqrt q-1)\leq\varepsilon\delta$, we seek
a better ratio $A/n$ between the number of evaluation maps and the bound
on their total rank loss. For this, we want more evaluation places.
Places of higher degree can be much more abundant than rational places,
so we consider evaluating the functions $\tau_i(f)$ at places $P$ of
degree $\ell\geq2$. The resulting maps are
$E_P\colon V\to\FF_{q^\ell}^r$. However, representing these outputs
over $\FF_q$ gives output dimension $t=\ell r$, whereas we want
$t=r$. To retain this output dimension, we adapt an idea of Ta-Shma
and Umans \cite{TSU12} and apply a second polynomial evaluation:
we view each output coordinate as a low-degree polynomial over $\FF_q$
and evaluate it at a point over $\FF_q$. We now describe the choice
of places and polynomials for these two steps.

For the first evaluation, we seek for a tower of function fields over $\FF_q$ with
many degree-$\ell$ places. We obtain it from the Garcia--Stichtenoth
tower over $\FF_{q^\ell}$; its defining equations have coefficients
in $\FF_q$, and hence can be considered directly over $\FF_q$.
We refer to this tower as the \emph{$\ell$-descended GS tower}. In \cref{cor:degree-d-places-gs-tower}, we show that this tower
has many degree-$\ell$ places (relative to its genus). More precisely, it contains $N$
degree-$\ell$ places satisfying
\[
    N\geq\frac{q^{\ell/2}-1}{2\ell}\,g.
\]
Moreover, we show that the required function spaces, evaluation places and the large-degree place $Q$ can also be computed efficiently. By restricting $\ell$ to be an even integer, we guarantee that $q^\ell$ is a perfect square regardless of $q$. This bypasses the previous requirement that the base field itself be of square order.

For the second evaluation, we identify $\FF_{q^\ell}$ with
$\FF_q^\ell$, and view each vector $b=(b_0,\ldots,b_{\ell-1})$ as the affine-linear
polynomial
\[
    g_b(\boldsymbol{X})
    =b_0+\sum_{j=1}^{\ell-1}b_jX_j.
\]
Here $\boldsymbol{X}=(X_1,\ldots,X_{\ell-1})$. For each
$\boldsymbol{a}\in\FF_q^{\ell-1}$, we evaluate the polynomials
representing all coordinates of $E_P(f)$ at the same point
$\boldsymbol{a}$. This gives an $\FF_q$-linear map
$E_{P,\boldsymbol{a}}\colon V\to\FF_q^r$, and hence a collection of
$n=Nq^{\ell-1}$ maps. Using linear multivariate polynomials
rather than univariate polynomials of degree $ < \ell$,
 minimizes the degree of each polynomial, although at the cost of additional evaluation points.
In our parameter regime of small $q$ and $\ell$, an overhead of $q^{\ell - 1}$ evaluation points
is sufficiently small, making this trade-off attractive.

In \cref{thm:arbitrary-field-strong-subspace-designs}, we show that,
for every prime power $q$, every even $\ell\geq2$
and every $1\leq r\leq k$, this construction gives an explicit $(r, A)$-subspace design,
where
\[
    \begin{gathered}
    A=\left(\frac{1}{q} +O(q^{-\ell/2}) \right)
      rn,
    \\
    n\leq q^{O(\ell)} \cdot k.
    \end{gathered}
\]
Dually, we obtain an $(r,\varepsilon,\delta)$-lossy rank condenser of size $n$,
provided that $\frac{1}{q}+O(q^{-\ell/2})\leq\varepsilon\delta$.
The bound on the average relative rank loss approaches $1/q$ as
$\ell$ grows.

For any $0 < \gamma < 1$, we can therefore choose an even integer $\ell$ such that the average rank loss is less than $(1/q+\gamma)r$, yielding a collection of $n = O(q^3k/\gamma^5)$ maps. These maps constitute an $(r, \varepsilon, \delta)$-lossy rank condenser provided that $1/q+\gamma \leq \varepsilon\delta$. Consequently, the construction applies to every finite field, including $\FF_2$; see \cref{cor:all-field-lossy-rank-condensers}.

\section{Preliminaries}

Throughout, $\FF$ denotes an arbitrary field, and $\FF_q$ denotes the finite field
with $q$ elements. We write $\log=\log_2$ unless a different base is specified.
We view a matrix $E\in\FF^{t\times k}$ as
the linear map $x\mapsto Ex$ from $\FF^k$ to $\FF^t$, and use
$\operatorname{rank}(E)$ and $\ker(E)$ for its rank and kernel,
respectively.

We will use the following form of the reduction of Forbes and Guruswami \cite[Proposition 6.7]{FG15} from strong subspace designs to lossy rank condensers. A similar strengthening was also used recently in \cite{DG26}. We include the proof here for completeness.

\begin{proposition}[Strong subspace designs as lossy rank condensers]
\label{prop:strong-design-to-lossy-condenser}
Let $\FF$ be a field, let $1\leq r\leq t\leq k$, let $n\geq1$ and $A\geq0$
be integers, and let $\varepsilon,\delta\in[0,1]$. Let
$\mathcal{H}=(H_i)_{i\in[n]}$ be an $(r,A)$-strong subspace design in
$\FF^k$ of size $n$ and codimension $t$. For every $i\in[n]$, let
$E_i\in\FF^{t\times k}$ satisfy $\ker(E_i)=H_i$, and set
$\mathcal{E}=(E_i)_{i\in[n]}$. If
\[
    A<\bigl(\lfloor\delta n\rfloor+1\bigr)
       \bigl(\lfloor\varepsilon r\rfloor+1\bigr),
\]
then $\mathcal{E}$ is an $(r,\varepsilon,\delta)$-lossy rank condenser.
\end{proposition}
\begin{proof}
Fix a rank-$r$ matrix $M\in\FF^{k\times r}$, and let
$W\leq\FF^k$ be its column space. The columns of $M$ form a basis of $W$.
Thus, for every $i\in[n]$, the matrix $E_iM$ represents the restriction of
$E_i$ to $W$. We have
\[
    \operatorname{rank}(E_iM)
    =r-\dim_{\FF}(W\cap H_i).
\]
Let
\[
    \mathcal{B}
    =\{i\in[n]:\operatorname{rank}(E_iM)<(1-\varepsilon)r\}
\]
be the set of bad indices. For every $i\in\mathcal{B}$, we have
\[
    \dim_{\FF}(W\cap H_i)\geq\lfloor\varepsilon r\rfloor+1.
\]
The strong subspace design property therefore implies
\[
    A
    \geq\sum_{i=1}^n\dim_{\FF}(W\cap H_i)
    \geq|\mathcal{B}|\bigl(\lfloor\varepsilon r\rfloor+1\bigr).
\]
The assumed bound on $A$ now gives
\[
    |\mathcal{B}|\leq\lfloor\delta n\rfloor.
\]
Thus, at most $\delta n$ indices are bad, as needed.
\end{proof}
\begin{remark}
    For $\delta=1-1/n$, the hypothesis in the proposition becomes
    \[
        n>\frac{A}{\lfloor\varepsilon r\rfloor+1},
    \]
    which recovers \cite[Proposition 6.7]{FG15}.
\end{remark}

For $0\neq f\in\FF[\boldsymbol{X}]$ and
$\boldsymbol{a}\in\FF^t$, let
$\operatorname{mult}_{\boldsymbol{a}}(f)$ be the least total degree of a
nonzero monomial in $f(\boldsymbol{a}+\boldsymbol{Y})$.
We make use of a variant of the Schwartz--Zippel lemma due to Dvir et al.~\cite[Lemma~8]{DZK09}, which bounds the number of zeroes with multiplicity.

\begin{lemma}[Multiplicity Schwartz--Zippel bound]
\label{lem:multiplicity-schwartz-zippel}
Let $t\geq1$, let $S$ be a finite subset of a field $\FF$, and let
$0\neq f\in\FF[X_1,\ldots,X_t]$ have total degree at most $D$. Then
\[
    \sum_{\boldsymbol{a}\in S^t}
    \operatorname{mult}_{\boldsymbol{a}}(f)
    \leq D|S|^{t-1}.
\]
\end{lemma}

We will use the existence of prime numbers in large enough intervals, as guaranteed by Bertrand's postulate.
\begin{fact}[Bertrand's postulate]\label{fact:Bertrand}
For every integer $n>1$, there exists a prime $p$ such that $n<p < 2n$.
\end{fact}

\subsection{Function fields}\label{subsec:function-field-prelim}

We assume familiarity with the basic theory of algebraic function fields, such as places, valuations, extensions, etc. A detailed exposition can be
found in \cite{Stich}. We briefly recall the relevant definitions and results.

\paragraph{Algebraic function fields, valuations, and places.}
Let $x$ be an indeterminate over $\FF_q$. The rational function field
$\FF_q(x)$ consists of the rational functions in $x$ with coefficients in
$\FF_q$. An algebraic function field $F/\FF_q$ is a finite algebraic
extension of $\FF_q(x)$, and its elements are called functions. Throughout,
we assume that $\FF_q$ is the full constant field of $F$; that is, every
element of $F$ that is algebraic over $\FF_q$ belongs to $\FF_q$.

A discrete valuation on $F$ is a map $v\colon F^\times\to\ZZ$ satisfying
\[
    v(fg)=v(f)+v(g)
    \qquad\text{and}\qquad
    v(f+g)\geq\min\{v(f),v(g)\},
\]
and we extend it to $F$ by setting $v(0)=\infty$. We consider valuations that
are trivial on $\FF_q^\times$ and normalized, meaning that
$v(F^\times)=\ZZ$. Associated with such a valuation are its valuation ring,
maximal ideal, and residue field,
\[
    \mathcal{O}_v=\{f\in F:v(f)\geq0\},\qquad
    \mathfrak{m}_v=\{f\in F:v(f)>0\},\qquad
    \kappa(v)=\mathcal{O}_v/\mathfrak{m}_v,
\]
respectively. A place $P$ of $F$ is the maximal ideal associated with a
normalized discrete valuation. We write $v_P$, $\mathcal{O}_P$, and
$\mathfrak{m}_P$ for the corresponding valuation, valuation ring, and
maximal ideal, and we write $\kappa(P)=\mathcal{O}_P/\mathfrak{m}_P$ for the
residue field.

The degree of a place $P$ is
\[
    \deg P=[\kappa(P):\FF_q].
\]
For $f \in \mathcal{O}_P$, we denote the residue class of $f$ modulo
$\mathfrak{m}_P$ by $f(P) \in \kappa(P)$. A place of degree one is called
\emph{rational}; for a rational place $P$, we have $\kappa(P) = \FF_q$
and thus $f(P) \in \FF_q$.

\paragraph{Divisors, principal divisors, and pole divisors.}
A divisor of $F$ is a finite formal sum
\[
    D=\sum_P n_PP,\qquad n_P\in\ZZ.
\]
Its support and degree are
\[
    \operatorname{supp}(D)=\{P:n_P\neq0\},
    \qquad
    \deg D=\sum_P n_P\deg P.
\]
For divisors $D_1$ and $D_2$, we write $D_1\geq D_2$ if the coefficient of
every place in $D_1$ is at least the corresponding coefficient in $D_2$.
A divisor $D$ is effective if $D\geq0$.

Every nonzero function $f\in F^\times$ determines a divisor
\[
    (f)=\sum_Pv_P(f)P.
\]
The divisor $(f)$ is called the principal divisor of $f$.
Its zero divisor and pole divisor are, respectively,
\[
    (f)_0=\sum_{P:v_P(f)>0}v_P(f)P,
    \qquad
    (f)_\infty=\sum_{P:v_P(f)<0}-v_P(f)P.
\]
Thus, $(f)=(f)_0-(f)_\infty$. We will repeatedly use the following standard
fact about principal divisors.

\begin{lemma}[Degree of a principal divisor]\label{lem:principal-divisor-degree}
For every $f\in F^\times$,
\[
    \deg((f))=0
    \qquad\text{and hence}\qquad
    \deg((f)_0)=\deg((f)_\infty).
\]
\end{lemma}

\paragraph{Riemann--Roch spaces and the genus.}
For a divisor $D$, its Riemann--Roch space is
\[
    \mathcal{L}(D)
    =\{f\in F^\times:(f)+D\geq0\}\cup\{0\}.
\]
This is a finite-dimensional $\FF_q$-vector space, and we write
$\ell(D)=\dim_{\FF_q}\mathcal{L}(D)$. Thus, $\ell(D)$ measures the number
of linearly independent functions whose poles are bounded by $D$. A direct corollary of \cref{lem:principal-divisor-degree} is the following lemma.

\begin{lemma}[Negative-degree divisors]\label{lem:negative-degree-divisor}
If $D$ is a divisor of $F$ with $\deg D<0$, then
\[
    \mathcal{L}(D)=\{0\}.
\]
\end{lemma}

\begin{theorem}[Riemann--Roch]
For every function field $F/\FF_q$, there exists a unique nonnegative integer
$g=g(F)$ such that
\[
    \ell(D)\geq\deg D-g+1
\]
for every divisor $D$ of $F$, with equality whenever $\deg D\geq2g-1$.
The integer $g$ is called the genus of $F$.
\end{theorem}

\paragraph{Bounds on the number of rational places.}
The Hasse--Weil theorem bounds the number of rational places in terms of the
genus and the size of the constant field.

\begin{theorem}[Hasse--Weil]\label{thm:hasse-weil}
Let $F/\FF_q$ be a function field of genus $g$. Then
\[
    \bigl|N(F)-(q+1)\bigr|\leq2g\sqrt q.
\]
\end{theorem}

We will also use a consequence of the Hasse--Weil theorem for places of
higher degree.

\begin{corollary}[Number of degree-$r$ places]\label{cor:number-degree-r-places}
Let $F/\FF_q$ be a function field of genus $g$, let $r\geq1$, and let
$B_r=B_r(F)$ denote the number of places of $F$ of degree $r$. Then,
\begin{equation*}
    \bigl|B_r-\frac{q^r}{r}\bigr|
    < (2+7g)\cdot \frac{q^{r/2}}{r}.
\end{equation*}
\end{corollary}

Thus, for a fixed field $\FF_q$, the number of rational places grows at most
linearly with the genus. To describe the best possible asymptotic ratio, we
use Ihara's constant.

\begin{definition}[Ihara's constant]
For every prime power $q$, let
\[
    N_q(g)=\max\{N(F):F/\FF_q\text{ is a function field of genus }g\}
\]
and define
\[
    A(q)=\limsup_{g\to\infty}\frac{N_q(g)}{g}.
\]
The quantity $A(q)$ is called \emph{Ihara's constant}.
\end{definition}

\begin{theorem}[Drinfeld--Vladut bound]\label{thm:drinfeld-vladut}
For every prime power $q$,
\[
    A(q)\leq\sqrt q-1.
\]
\end{theorem}

\subsubsection{The Garcia--Stichtenoth tower}
\label{subsec:gs-tower-background}

When $q$ is a square, the Drinfeld--Vladut bound is attained by the Garcia-Stichtenoth tower of function fields.
A detailed exposition of this tower can be found in \cite{Garcia1995}. We briefly recall the relevant definitions and results.

\begin{definition}[Garcia--Stichtenoth tower]\label{def:gs-tower}
Let $\FF_q$ be a finite field of size $q=\ell^2$, where $\ell$ is a prime power.
The \emph{Garcia--Stichtenoth tower} is the sequence of function fields
\[
F_1 \subseteq F_2 \subseteq \cdots
\]
defined recursively by
\[
F_1=\FF_q(x_1),
\]
and for every $i\ge1$,
\[
F_{i+1}=F_i(x_{i+1}),
\]
where the new variable $x_{i+1}$ satisfies the equation
\[
x_{i+1}^{\ell}+x_{i+1}
=\frac{x_i^{\ell}}{x_i^{\ell-1}+1}.
\]
\end{definition}

The following theorem records the parameters of this tower that we use in
the construction.

\begin{theorem}[The Garcia--Stichtenoth tower]
\label{thm:gs-tower-parameters}
Let $(F_i)_{i\geq1}$ be the Garcia--Stichtenoth tower over $\FF_q$, where
$q=\ell^2$. For $i\geq1$, let $N_i=N(F_i)$ and $g_i=g(F_i)$. Then
\[
    [F_i:F_1]=\ell^{i-1},
    \qquad
    N_i\geq\ell^i(\ell-1)+1,
\]
and
\[
g_i=
\begin{cases}
\left(\ell^{i/2}-1\right)^2, & \text{if } i \text{ is even},\\[1ex]
\left(\ell^{(i+1)/2}-1\right)
\left(\ell^{(i-1)/2}-1\right), & \text{if } i \text{ is odd}.
\end{cases}
\]
In particular, $g_i\leq\ell^i$ and
\[
    \limsup_{i\to\infty}\frac{N_i}{g_i}
    \geq\ell-1=\sqrt q-1.
\]
\end{theorem}

We record two further features of the tower that will be useful later.

\paragraph{Rational places.}
At every level, let $P_\infty$ denote the unique place above the pole of $x_1$
in $F_1$. This place is rational and totally ramified throughout the tower.
For each $\alpha_1\in\FF_q$ satisfying $\alpha_1^\ell+\alpha_1\neq0$, the zero of
$x_1-\alpha_1$ in $F_1$ splits completely in $F_i$. Indeed, these places
may be represented by the tuples
$
    (\alpha_1,\ldots,\alpha_i)\in\FF_q^i
$
such that $\alpha_1^\ell+\alpha_1\neq0$ and
\[
    \alpha_{j+1}^\ell+\alpha_{j+1}
    =\frac{\alpha_j^\ell}{\alpha_j^{\ell-1}+1}
    \qquad (1\leq j<i).
\]
There are $q-\ell$ possible values of $\alpha_1$ and $\ell$ choices for each
subsequent coordinate. We denote the set of these places by $\mathcal{R}_i$. We have
\begin{equation}\label{eq:explicit-places-gs}
    |\mathcal{R}_i| = \ell^i (\ell - 1).
\end{equation}
Together with $P_\infty$, these places give the lower bound on $N_i$ in
\cref{thm:gs-tower-parameters}.

\paragraph{Riemann--Roch spaces.}
For $i\geq2$, define
\[
    h_j=x_j^{\ell-1}+1,
    \qquad
    \pi_j=h_1h_2\cdots h_j
    \qquad (1\leq j<i).
\]
Every function in $F_i$ whose poles are supported at $P_\infty$ admits an
expression of the form
\[
        \sum_{\nu_2=0}^{\ell-1}\cdots
        \sum_{\nu_i=0}^{\ell-1}
        p_{\boldsymbol{\nu}}(x_1)
        \frac{x_2^{\nu_2}x_3^{\nu_3}\cdots x_i^{\nu_i}}
             {\pi_1\pi_2\cdots\pi_{i-1}}, \qquad
        p_{\boldsymbol{\nu}} \in \FF_q[x_1].
\]
Moreover, the algorithm of Shum et al.~\cite{ShumEtAl2001} computes an
explicit $\FF_q$-basis of $\mathcal{L}(\mu P_\infty)$ in this form in
time $\text{poly}(\log q,g_i, \mu)$.

For every $P\in\mathcal{R}_i$, all the values $h_j(P)$ and $\pi_j(P)$ are
nonzero. Hence the basis functions, and therefore every element of
$\mathcal{L}(\mu P_\infty)$ given in this basis, can be evaluated at the
places in $\mathcal{R}_i$ by substituting their coordinate tuples. In
particular, these evaluations can be computed in time $\text{poly}(\log q, \mu, N_i)$.

\paragraph{Constant-field extensions and places of higher degree.}

Let $F/\FF_q$ be a function field with full constant field $\FF_q$. For an
integer $a\geq1$, we write
\[
    F^{(a)}=F\cdot\FF_{q^a}
\]
for its constant-field extension to $\FF_{q^a}$. The following lemma relates
the rational places of $F^{(a)}$ to the places of $F$.

\begin{lemma}[Places under a constant-field extension]
\label{lem:places-under-constant-extension}
The function field $F^{(a)}/\FF_{q^a}$ has the same genus as $F/\FF_q$.
Moreover, a place of $F$ of degree $d$ splits in $F^{(a)}$ into
$\gcd(d,a)$ places, each of degree $d/\gcd(d,a)$. Consequently,
\begin{equation}\label{eq:rational-places-constant-extension}
    N\bigl(F^{(a)}\bigr)=\sum_{d\mid a}dB_d(F).
\end{equation}
\end{lemma}

In \cref{sec:lossless-rank-extractors} we apply the lemma to a Garcia--Stichtenoth tower over constant field extensions to obtain a family of function fields with many places of a given degree.

We finally record a local linear-algebra estimate used in the constructions.

\begin{lemma}[Kernel dimension and determinant valuation]
\label{lem:kernel-dimension-determinant-valuation}
Let $P$ be a place of a function field $F/\FF_q$, and let
$A\in\mathcal{O}_P^{s\times s}$  have nonzero determinant. Denote the
reduction of $A$ modulo $P$ coordinate-wise by
$A(P)\in\kappa(P)^{s\times s}$. Then
\[
    \dim_{\kappa(P)}\ker A(P)
    \leq v_P(\det A).
\]
\end{lemma}
\begin{proof}
Set $u=\dim_{\kappa(P)}\ker A(P)$. Choose
$\overline{C}\in\operatorname{GL}_s(\kappa(P))$ whose first $u$ columns
form a basis of $\ker A(P)$, and lift $\overline{C}$ to a matrix
$C\in\mathcal{O}_P^{s\times s}$.  The reduction of $\det C$ modulo $P$ is
$\det\overline{C}\neq0$, so $\det C$ is a unit in $\mathcal{O}_P$. Hence
\[
    v_P(\det(AC))=v_P(\det A).
\]
The first $u$ columns of $AC$ vanish modulo $P$. Every term in the Leibniz
formula for $\det(AC)$ therefore has valuation at least $u$, and hence
$v_P(\det A)\geq u$.
\end{proof}

\section{Algebraic-Geometric Parvaresh--Vardy Subspace Designs}\label{sec:agpv construction}
\subsection{Construction framework}\label{subsec:construction-framework}
Let $F/\FF_q$ be a function field with full constant field $\FF_q$,
genus $g$, and $N$ rational places. Fix integers $1\leq r\leq k$, and assume
that $g\leq k$. Set
\[
    d=k+g-1.
\]
Suppose that $F$ contains a rational place $P_\infty$ and a place
$Q\neq P_\infty$ of degree $m$, where $d<m<2d$.

Let $V=\mathcal{L}(dP_\infty)$. Since $d\geq2g-1$, the Riemann--Roch theorem
gives
\[
    \dim_{\FF_q}V=d+1-g=k.
\]
Write
$\kappa(Q)=\mathcal{O}_Q/\mathfrak{m}_Q\cong\FF_{q^m}$ for the residue field
of $Q$, and let $\pi_Q\colon\mathcal{O}_Q\to\kappa(Q)$ be the residue map. Its
restriction to $V$ is injective, since
\[
    \ker(\pi_Q|_V)=\mathcal{L}(dP_\infty-Q)=\{0\},
\]
by $\deg(dP_\infty-Q)=d-m<0$. This map is not necessarily surjective, so we now
lift it to a larger space.

Let $\widetilde d=m+2g-1$, and consider the space
\[
    \widetilde V=\mathcal{L}(\widetilde dP_\infty).
\]
The restriction of $\pi_Q$ to $\widetilde V$ \emph{is} surjective. Indeed, its kernel is
$\mathcal{L}(\widetilde dP_\infty-Q)$. Since
$\deg(\widetilde dP_\infty-Q)=2g-1$, the Riemann--Roch theorem gives
\begin{align*}
    &\dim_{\FF_q}\widetilde V = m + g,\\
    &\dim_{\FF_q}\ker(\pi_Q|_{\widetilde V})=g.
\end{align*}
Thus, the image has dimension $m=\deg Q=\dim_{\FF_q}\kappa(Q)$. Fix an
$\FF_q$-linear right inverse
\[
    I\colon\kappa(Q)\longrightarrow\widetilde V
    \qquad\text{such that}\qquad
    \pi_Q\circ I=\operatorname{id}_{\kappa(Q)}.
\]

For every $0\leq i\leq r-1$, define an
$\FF_q$-linear map $\tau_i\colon V\to\widetilde V$ by
\[
    \tau_0(f)=f,
    \qquad
    \tau_i(f)=I\bigl(\pi_Q(f)^{q^i}\bigr)
    \quad\text{for }i\geq1.
\]
These maps satisfy
\begin{equation}\label{eq:tau-mod-Q}
    \pi_Q(\tau_i(f))=\pi_Q(f)^{q^i},
\end{equation}
for every $f\in V$ and $0\leq i\leq r-1$.

\begin{definition}[AGPV subspace design]\label{def:agpv-sd}
Let $F/\FF_q$ be a function field as above. Let $\mathcal{P}$ be a set of rational places of $F$ different from
$P_\infty$. The functions in $\widetilde V$ are regular at every place in
$\mathcal{P}$. Thus, for each $P\in\mathcal{P}$, we may define the
$\FF_q$-linear map
\[
    E_P\colon V\longrightarrow\FF_q^r,
    \qquad
    E_P(f)=\bigl(\tau_0(f)(P),\ldots,\tau_{r-1}(f)(P)\bigr),
\]
and let
\[
    H_P=\ker E_P \leq V.
\]
Each $H_P$ has codimension at most $r$ in $V$. We define the indexed
collection
\[
    \mathcal{H}=\mathcal{H}(F, \mathcal{P}, P_\infty, Q, k, r) = (H_P)_{P\in\mathcal{P}},
\]
which has size $n = |\mathcal{P}|$.
\end{definition}

\subsection{The subspace-design property}

In the following part, we show that the collection $\mathcal{H}$ is a strong subspace design.

In our analysis, we will make use of a generalization of the Moore Matrix, presented by E.H. Moore~\cite{Moore1896}.
\begin{definition}[Moore Matrix]
    Let $F/\FF_q$ be a field extension. Let $a_1,\ldots, a_s$ be elements of $F$.
    Their associated \emph{Moore matrix} is defined as
    \[
        M(a_1,\ldots,a_s)
        =
        \begin{pmatrix}
            a_1 & a_2 & \cdots & a_s \\
            a_1^q & a_2^q & \cdots & a_s^q \\
            \vdots & \vdots & \ddots & \vdots \\
            a_1^{q^{s-1}} & a_2^{q^{s-1}} & \cdots & a_s^{q^{s-1}}
        \end{pmatrix}
        \in F^{s\times s}.
    \]
\end{definition}
A standard fact is that the elements $a_1,\ldots,a_s$ are linearly independent over $\FF_q$ if and only if $M(a_1,\ldots,a_s)$ is invertible over $F$.
The following definition generalizes the Moore matrix to the lifted Frobenius maps $\tau_i$.

\begin{definition}[PV-Moore matrix]
Let $1\leq s\leq r$, and let
$\mathcal{B}=(f_1,\ldots,f_s)$ be an ordered tuple of elements of $V$ as defined above. We
define its \emph{Parvaresh--Vardy Moore (PV-Moore) matrix} by
\[
    M_{\mathcal{B}}
    =
    \begin{pmatrix}
        \tau_0(f_1) & \cdots & \tau_0(f_s) \\
        \tau_1(f_1) & \cdots & \tau_1(f_s) \\
        \vdots      & \ddots & \vdots      \\
        \tau_{s-1}(f_1) & \cdots & \tau_{s-1}(f_s)
    \end{pmatrix}
    \in F^{s\times s}.
\]
\end{definition}

The usual Moore-matrix criterion remains valid for the lifted Frobenius maps. We include the proof here for completeness.

\begin{proposition}\label{prop:invertible-pv-moore}
Let $1\leq s\leq r$, and let
$\mathcal{B}=(f_1,\ldots,f_s)$ be an ordered tuple of elements of $V$. Then
$\mathcal{B}$ is linearly independent over $\FF_q$ if and only if
$M_{\mathcal{B}}$ is invertible over $F$.
\end{proposition}
\begin{proof}
If $\mathcal{B}$ is linearly dependent over $\FF_q$, then the $\FF_q$-linearity
of the maps $\tau_i$ gives a nontrivial linear dependence among the columns of
$M_{\mathcal{B}}$.

Conversely, suppose that $\mathcal{B}$ is linearly independent over $\FF_q$.
Since $\pi_Q$ is injective on $V$, the elements
\[
    \alpha_j=\pi_Q(f_j),\qquad 1\leq j\leq s,
\]
are linearly independent over $\FF_q$. By \cref{eq:tau-mod-Q}, reducing
$M_{\mathcal{B}}$ modulo $Q$ gives the Moore matrix
\[
    M_{\mathcal{B}}(Q)
    :=\bigl(\alpha_j^{q^i}\bigr)_{
        0\leq i\leq s-1,\,1\leq j\leq s}
    \in\kappa(Q)^{s\times s}.
\]
If this matrix were singular, there would exist
$a_0,\ldots,a_{s-1}\in\kappa(Q)$, not all zero, such that the nonzero
$\FF_q$-linear polynomial
\[
    L(X)=\sum_{i=0}^{s-1}a_iX^{q^i}
\]
vanishes at every $\alpha_j$. It would then vanish on their $\FF_q$-span,
which has $q^s$ distinct elements. This is impossible because
$\deg L\leq q^{s-1}$. Hence ${M}_{\mathcal{B}}(Q)$ is invertible, which implies that $M_{\mathcal{B}}$ is invertible.
\end{proof}

Next, we show that the dimension of the intersection of a subspace $W\leq V$ with $H_P$ is closely related to the determinant of the PV-Moore matrix of a basis of $W$.

\begin{lemma}\label{lemma:PV-Moore-det-valuations}
Let $W\leq V$ be a subspace of dimension $s$, where $1\leq s\leq r$, and let
$\mathcal{B}=(f_1,\ldots,f_s)$ be an $\FF_q$-basis of $W$. Then, for every
$P\in\mathcal{P}$,
\[
    \dim_{\FF_q}(W\cap H_P)
    \leq v_P(\det M_{\mathcal{B}}).
\]
\end{lemma}
\begin{proof}
Evaluating the entries of $M_{\mathcal{B}}$ at $P$ gives
\[
    M_{\mathcal{B}}(P)
    =\bigl(\tau_i(f_j)(P)\bigr)_{
        0\leq i\leq s-1,\,1\leq j\leq s}
    \in\FF_q^{s\times s}.
\]
Every $f\in W$ can be written uniquely as
\[
    f=\sum_{j=1}^s c_jf_j,
    \qquad c_j\in\FF_q.
\]
Since the maps $\tau_i$ are $\FF_q$-linear, $f \in W\cap H_P$ implies
\[
    M_{\mathcal{B}}(P)(c_1,\ldots,c_s)^T=0.
\]
Consequently,
\[
    \dim_{\FF_q}(W\cap H_P)
\leq \dim_{\FF_q}\ker M_{\mathcal{B}}(P).
\]
The matrix $M_{\mathcal{B}}$ has entries in $\mathcal{O}_P$, and its
determinant is nonzero by \cref{prop:invertible-pv-moore}. The result now
follows from \cref{lem:kernel-dimension-determinant-valuation}.
\end{proof}

We are ready to prove that $\mathcal{H}$ is a strong subspace design.

\begin{proposition}\label{prop:strong-subspace-design}
For every subspace $W\leq V$ of dimension $1\leq s\leq r$,
\[
    \sum_{P\in\mathcal{P}}\dim_{\FF_q}(W\cap H_P)
    \leq d+(s-1)\widetilde d
    \leq s\widetilde d.
\]
Consequently, $\mathcal{H}$ is an
$(r,d+(r-1)\widetilde d)$-strong subspace design of size $|\mathcal{P}|$ whose members have
codimension at most $r$. In particular, $\mathcal{H}$ is an
$(r,6rk)$-strong subspace design.
\end{proposition}

\begin{proof}
Fix a subspace $W\leq V$ of dimension $1\leq s\leq r$, and let
$\mathcal{B}$ be an $\FF_q$-basis of $W$. By
\cref{prop:invertible-pv-moore}, the determinant of $M_{\mathcal{B}}$ is
nonzero. The entries in its first row belong to $\mathcal{L}(dP_\infty)$,
whereas the entries in each of the remaining $s-1$ rows belong to
$\mathcal{L}(\widetilde dP_\infty)$. Therefore
\[
    0\neq\det M_{\mathcal{B}}
    \in\mathcal{L}((d+(s-1)\widetilde d)P_\infty).
\]
Together with \cref{lemma:PV-Moore-det-valuations}, we obtain
\begin{align*}
    \sum_{P\in\mathcal{P}}\dim_{\FF_q}(W\cap H_P)
    &\leq\sum_{P\in\mathcal{P}}v_P(\det M_{\mathcal{B}})\\
    &\leq\deg((\det M_{\mathcal{B}})_0)\\
    &=\deg((\det M_{\mathcal{B}})_\infty)\\
    &\leq d+(s-1)\widetilde d.
\end{align*}
Since $d\leq\widetilde d$, the latter quantity is at most
$s\widetilde d$. Together,
\[
    \widetilde d=m+2g-1
    <2d+2g-1
    =2k+4g-3
    <6k,
\]
which gives the final assertion.
\end{proof}

\subsection{Instantiating with the Garcia--Stichtenoth function field}

Recall that our construction requires explicitly finding a place $Q$ of degree $m$
satisfying $d < m < 2d$. The following proposition guarantees the existence
and efficient construction of such a place for our instantiation.

\begin{proposition}[Explicit place of prime degree]
\label{prop:explicit-prime-degree-place}
Let $q$ be a prime power, let $\ell\geq1$, and suppose that $q^\ell=r^2$ for a
prime power $r$. Let $F_i/\FF_q$ be the $i$'th field in the $\ell$-descended
Garcia--Stichtenoth tower, defined by
\[
    F_1=\FF_q(x_1),\qquad
    F_{j+1}=F_j(x_{j+1}),\qquad
    x_{j+1}^{r}+x_{j+1}
    =\frac{x_j^{r}}{x_j^{r-1}+1}
    \qquad (j\geq1),
\]
and suppose that $F_i$ has genus $g\geq2$. Given a prime number $m$ with
$2g-1<m\leq 2(2g-1)$, there is a deterministic $\poly(q,\ell,g)$-time algorithm
that outputs an explicit
representation $K\cong\FF_{q^m}$ and elements
$a_1,\ldots,a_i\in K$ such that the evaluation $x_j\mapsto a_j$ represents
a degree-$m$ place $Q$ of $F_i/\FF_q$. In particular, the residue map at
$Q$ is explicit.
\end{proposition}
\begin{proof}
Suppose first that $\FF_q$ has odd characteristic. Fix such a prime $m$, and write
\[
    \phi(X)=\frac{X^r}{X^{r-1}+1}.
\]
Write $q=p^e$ and $r=p^s$, so that $e\ell=2s$. After extending the constant
field to $\FF_{q^\ell}=\FF_{r^2}$, the tower is the Garcia--Stichtenoth tower
of \cref{def:gs-tower}. Its genus is $g$, and
\cref{thm:gs-tower-parameters} gives
\[
    g\geq(r-1)^2\geq s.
\]
Since $m$ is an odd prime and $m>2g-1\geq2s-1$, we have $m>2s\geq e$.
Consequently,
\[
    \gcd(m,s)=\gcd(m,e)=1.
\]

Using Shoup's algorithm
\cite{Shoup1990}, deterministically construct an irreducible polynomial
$h\in\FF_p[T]$ of degree $m$. Since $\gcd(m,e)=1$, $h$ remains irreducible
over $\FF_q$. We may therefore set
\[
    L=\FF_p[T]/(h),\qquad K=\FF_q[T]/(h),
\]
and let $u=T+(h)\in L$, viewed also as an element of $K$. Then
$L\cong\FF_{p^m}$ is a subfield of $K\cong\FF_{q^m}$, and $u$ generates
both $L/\FF_p$ and $K/\FF_q$. The $\FF_p$-linear map
$A\colon L\to L$ given by $A(z)=z^r+z$ is invertible, with inverse given by
\[
    A^{-1}(b)=\frac{1}{2}\sum_{k=0}^{m-1}(-1)^k b^{r^k}.
\]

Set $a_1=u$ and, for $1\leq j<i$, recursively compute in $L$
\[
    a_{j+1}
    = A^{-1}(\phi(a_j)) =\frac{1}{2}\sum_{k=0}^{m-1}
        (-1)^k\phi(a_j)^{r^k}.
\]
If $a_j\neq0$, then the injectivity of $A$ gives
\[
    a_j^{r-1}+1=\frac{A(a_j)}{a_j}\neq0.
\]
Thus $\phi(a_j)$ is defined and nonzero, and hence so is
$a_{j+1}=A^{-1}(\phi(a_j))$. By induction, all coordinates are nonzero and
none of the denominators vanishes.

The evaluation $x_j\mapsto a_j$ satisfies the defining equations of $F_i$, which defines a place $Q$. Since $a_1=u$
generates $K/\FF_q$, the residue field of $Q$ is $K$, and therefore
$\deg Q=m$. The field $K$, the elements $a_1,\ldots,a_i$, and the resulting
evaluation map can all be computed in time $\poly(q,\ell,g)$.

The characteristic $2$ case follows the same recursion, but requires a modified argument to find the first element. We defer this proof to \cref{app:function-fields-missing-proofs}.
\end{proof}

We now instantiate the construction with the Garcia--Stichtenoth tower to obtain explicit strong subspace designs.

\begin{theorem}\label{thm:square-field-subspace-design}
Let $q=\ell^2$, where $\ell$ is a prime power. There are infinitely many
integers $k$ such that, for every $1\leq r\leq k$, there exists an explicit
$(r,6rk)$-strong subspace design $\mathcal{H}$ in $\FF_q^k$ of codimension
$r$. Moreover, $\mathcal{H}$ is an $(s,6sk)$-strong subspace design for every
$1\leq s\leq r$. The integers $k$ may be taken to be the positive genera in
the Garcia--Stichtenoth tower. If $n=|\mathcal{H}|$ and
$W\leq\FF_q^k$ has dimension $s\leq r$, then
\[
    n\geq(\sqrt q-1)k
    \qquad\text{and}\qquad
    \mathbb{E}_{H\in\mathcal{H}}\dim_{\FF_q}(W\cap H)
    \leq\frac{6}{\sqrt q-1}\,\dim_{\FF_q}W.
\]
\end{theorem}
\begin{proof}
Let $F=F_e/\FF_q$ be the $e$'th field in the Garcia--Stichtenoth tower with $e \geq 2$, let
$k=g(F)$ and let $\mathcal{P} = \mathcal{R}_e$. Set
\[
    d=k+g-1=2g-1.
\]
By \cref{fact:Bertrand}, there is a prime number $m$ such that $d<m< 2d$; such a
prime can be found deterministically by enumeration. By
\cref{prop:explicit-prime-degree-place}, we can efficiently compute a place
$Q$ of degree $m$ in $F$.
We may therefore apply \cref{def:agpv-sd,prop:strong-subspace-design}. For
every $1\leq s\leq r$, the resulting collection is an $(s,6sk)$-strong
subspace design of size $n = |\mathcal{R}_e|$.

By \cref{eq:explicit-places-gs} we have,
\[
    n = \ell^e(\ell-1)\geq(\ell-1)k
    =(\sqrt q-1)k.
\]
Dividing the intersection bound by $n$ gives
\[
    \mathbb{E}_{H\in\mathcal{H}}\dim_{\FF_q}(W\cap H)
    \leq\frac{6}{\sqrt q-1}\,\dim_{\FF_q}W.
\]

Finally, explicitness follows from
\cref{prop:explicit-prime-degree-place,subsec:gs-tower-background}. The former
provides the place $Q$, its residue field, and the residue map $\pi_Q$, while
the latter provides the rational places, bases of the required Riemann--Roch
spaces, and their evaluations. The remaining maps and subspaces in the
construction are obtained by linear algebra over $\FF_q$.
\end{proof}

\section{An Improved, Double-Evaluation Construction}
\label{sec:lossless-rank-extractors}

The preceding construction has two shortcomings. First, its explicit
instantiation is restricted to square fields. Second, its normalized
intersection bound is $6/(\sqrt q-1)$ times the dimension of the test subspace.
Consequently, for an $r$-dimensional test subspace, this estimate yields
a nontrivial lossless rank extractor only when $\sqrt q-1>6r$, and hence does not
apply over fields of size less than 49.

In this section, we give a variant of the AGPV construction that works over all finite fields, and produces subspace designs with a better ratio between the intersection bound and the number of subspaces.
A direct consequence is the existence of explicit lossy rank condensers over \emph{every} finite field, and lossless rank extractors whenever $q > r$.
Building on the framework presented in \cref{sec:agpv construction}, we present the construction of subspace designs, and then derive the consequences for rank condensers and extractors.

\subsection{Subspace designs construction}
\label{subsec:strong-subspace-designs}

Let $\ell\geq2$ be even. Our construction has two steps.
The first step is the same as the AGPV construction, but we evaluate at places of degree $\ell$ instead of at rational places.
In the second step, we represent each residue-field element as an affine-linear polynomial in $\ell-1$ variables over $\FF_q$, and evaluate these polynomials at all points of $\FF_q^{\ell-1}$.
Evaluating these polynomials
at all points of $\FF_q^{\ell-1}$ yields a normalized intersection bound
that approaches $1/q$ as $\ell$ grows.
The use of a second polynomial evaluation follows the idea of \cite{TSU12}.

We begin by analyzing the second step of the construction, which is independent of the function field.

\begin{lemma}
\label{lem:affine-kernel-growth}
Let $s,t\geq1$ and $D\geq1$ be integers, let
$\boldsymbol{X}=(X_1,\ldots,X_t)$, and let
$B(\boldsymbol{X})\in\FF_q[\boldsymbol{X}]^{s\times s}$ have entries of
total degree at most $D$. Set
\[
    u=\dim_{\FF_q(\boldsymbol{X})}\ker B(\boldsymbol{X}).
\]
Then,
\[
    \sum_{\boldsymbol{a}\in\FF_q^t}
    \dim_{\FF_q}\ker B(\boldsymbol{a})
    \leq q^{t-1}\bigl(qu+D(s-u)\bigr).
\]
\end{lemma}
\begin{proof}
The assertion is immediate if $u=s$. Otherwise, choose an
$(s-u)\times(s-u)$ minor $C(\boldsymbol{X})$ of $B(\boldsymbol{X})$ whose
determinant $h(\boldsymbol{X})$ is nonzero. Then,
$\deg h\leq D(s-u)$.

For $\boldsymbol{a}\in\FF_q^t$, set
$v_{\boldsymbol{a}}=\dim_{\FF_q}\ker C(\boldsymbol{a})$. Since $C$ is a
minor of $B$, we have
\[
    \dim_{\FF_q}\ker B(\boldsymbol{a})-u
    \leq v_{\boldsymbol{a}}.
\]
Choose $R_{\boldsymbol{a}}\in\operatorname{GL}_{s-u}(\FF_q)$ whose last
$v_{\boldsymbol{a}}$ columns span $\ker C(\boldsymbol{a})$. These columns
of $C(\boldsymbol{X})R_{\boldsymbol{a}}$ vanish at $\boldsymbol{a}$, so
the determinant has multiplicity at least $v_{\boldsymbol{a}}$ there.
Since $\det R_{\boldsymbol{a}}$ is a nonzero constant, it follows that
\[
    v_{\boldsymbol{a}}
    \leq\operatorname{mult}_{\boldsymbol{a}}(h).
\]
Applying \cref{lem:multiplicity-schwartz-zippel} with $S=\FF_q$ gives
\[
    \sum_{\boldsymbol{a}\in\FF_q^t}
    \left(\dim_{\FF_q}\ker B(\boldsymbol{a})-u\right)
    \leq Dq^{t-1}(s-u).
\]
Adding $q^tu$ to both sides proves the result.
\end{proof}

We now define the construction. Recall the framework presented in \cref{subsec:construction-framework}.

\begin{definition}[Affine-evaluation AGPV construction]
\label{def:affine-evaluation-agpv-family} Let $F/\FF_q$ be a function field, let $\ell\geq2$ be even and let $\mathcal{P}_{\ell}$ be a set of
degree-$\ell$ places of $F$. Let $Q$ be the place given by \cref{subsec:construction-framework}. For each $P\in\mathcal{P}_{\ell}$, fix an
$\FF_q$-basis
$
    \{ 1,\omega_{P,1},\ldots,\omega_{P,\ell-1} \}
$
of $\kappa(P)$. For
\[
    b=b_0+\sum_{j=1}^{\ell-1}b_j\omega_{P,j}\in\kappa(P),
\]
define
\[
    g_{P,b}(X_1,\ldots,X_{\ell-1})
    =b_0+\sum_{j=1}^{\ell-1}b_jX_j.
\]
For $\boldsymbol{a}\in\FF_q^{\ell-1}$, let
$\lambda_{P,\boldsymbol{a}}\colon\kappa(P)\to\FF_q$ be the linear map
$b\mapsto g_{P,b}(\boldsymbol{a})$. Let
$
    E_{P,\boldsymbol{a}}\colon V\longrightarrow\FF_q^r
$ be defined by
\[
    E_{P,\boldsymbol{a}}(f)
    =\bigl(
       \lambda_{P,\boldsymbol{a}}(\tau_0(f)(P)),\ldots,
       \lambda_{P,\boldsymbol{a}}(\tau_{r-1}(f)(P))
     \bigr),
\]
and set $H_{P,\boldsymbol{a}}=\ker E_{P,\boldsymbol{a}}$. Finally, let
\begin{align*}
    \mathcal{H} = \mathcal{H}(F, \mathcal{P}_\ell, P_\infty, Q, k, r, \ell)
    &=\bigl(H_{P,\boldsymbol{a}}\bigr)_
      {P\in\mathcal{P}_{\ell},\,\boldsymbol{a}\in\FF_q^{\ell-1}}, \\
    \mathcal{E} = \mathcal{E}(F, \mathcal{P}_\ell, P_\infty, Q, k, r, \ell)
    &=\bigl(E_{P,\boldsymbol{a}}\bigr)_
      {P\in\mathcal{P}_{\ell},\,\boldsymbol{a}\in\FF_q^{\ell-1}}.
\end{align*}
The collections have size $n=q^{\ell-1}|\mathcal{P}_{\ell}|$.
\end{definition}

\begin{remark}
This construction introduces only one parameter beyond those of
\cref{def:agpv-sd}, namely $\ell$.
\end{remark}
We first analyze the construction over an arbitrary function field.

\begin{proposition}[Affine-evaluation kernel bound]
\label{prop:affine-evaluation-agpv}
In the setting of \cref{def:affine-evaluation-agpv-family}, for
$1\leq s\leq r$ let
\[
    A_s=q^{\ell-2}\left(
      sN+(q-1)\left\lfloor
        \frac{d+(s-1)\widetilde d}{\ell}
      \right\rfloor
    \right).
\]
Then, $\mathcal{H} = \mathcal{H}(F, \mathcal{P}_\ell, P_\infty, Q, k, r, \ell)$ is an
$(s,A_s)$-strong subspace design for every $1\leq s\leq r$.
\end{proposition}
\begin{proof}
Fix $1\leq s\leq r$, an $s$-dimensional subspace $W\leq V$, and an
$\FF_q$-basis $\mathcal{B}=(f_1,\ldots,f_s)$ of $W$. Set
\[
    D_s=d+(s-1)\widetilde d,
    \qquad
    b_s=\left\lfloor\frac{D_s}{\ell}\right\rfloor.
\]
For each $P\in\mathcal{P}_{\ell}$, define
\[
    B_P(\boldsymbol{X})
    =\left(
      g_{P,\tau_{i-1}(f_j)(P)}(\boldsymbol{X})
     \right)_{i,j\in[s]}, \qquad u_P=\dim_{\FF_q(\boldsymbol{X})}\ker B_P(\boldsymbol{X}).
\]
For $\boldsymbol{a}\in\FF_q^{\ell-1}$ and
$f=\sum_{j=1}^s c_jf_j\in W$, by the definition of
$H_{P,\boldsymbol{a}}$ we have $f\in H_{P,\boldsymbol{a}}$ implies
$
    B_P(\boldsymbol{a})(c_1,\ldots,c_s)^T=0,
$
and hence
\[
    \dim_{\FF_q}(W\cap H_{P,\boldsymbol{a}})
    \leq \dim_{\FF_q}\ker B_P(\boldsymbol{a}).
\]
Applying \cref{lem:affine-kernel-growth} therefore gives
\[
    \sum_{\boldsymbol{a}\in\FF_q^{\ell-1}}
    \dim_{\FF_q}(W\cap H_{P,\boldsymbol{a}})
    \leq q^{\ell-2}\bigl(s+(q-1)u_P\bigr).
\]

It remains to bound the sum of the kernel dimensions $u_P$. By
\cref{prop:invertible-pv-moore}, the determinant
$\Delta_{\mathcal{B}}:=\det M_{\mathcal{B}}$ is nonzero. Moreover, the first
row of $M_{\mathcal{B}}$ has entries in $\mathcal{L}(dP_\infty)$, while
each of the remaining $s-1$ rows has entries in
$\mathcal{L}(\widetilde dP_\infty)$. Therefore,
\[
    0\neq\Delta_{\mathcal{B}}\in\mathcal{L}(D_sP_\infty).
\]
For every $P\in\mathcal{P}_{\ell}$, we have
$
    B_P(\omega_{P,1},\ldots,\omega_{P,\ell-1})
    =M_{\mathcal{B}}(P),
$
and therefore
\[
    u_P\leq\dim_{\kappa(P)}\ker M_{\mathcal{B}}(P).
\]
Combining this with \cref{lem:kernel-dimension-determinant-valuation} gives
\[
    u_P\leq v_P(\Delta_{\mathcal{B}}).
\]
Since the places in $\mathcal{P}_{\ell}$ have degree $\ell$, we obtain
\[
    \ell\sum_{P\in\mathcal{P}_{\ell}}u_P
    \leq\ell\sum_{P\in\mathcal{P}_{\ell}}
      v_P(\Delta_{\mathcal{B}})
    \leq\deg((\Delta_{\mathcal{B}})_0)
    \leq D_s.
\]
Thus, $\sum_{P\in\mathcal{P}_{\ell}}u_P\leq b_s$. Summing the preceding
kernel bound over $P\in\mathcal{P}_{\ell}$ now gives
\begin{align*}
    \sum_{P\in\mathcal{P}_{\ell}}
    \sum_{\boldsymbol{a}\in\FF_q^{\ell-1}}
    \dim_{\FF_q}(W\cap H_{P,\boldsymbol{a}})
    &\leq q^{\ell-2}\left(
      sN+(q-1)\sum_{P\in\mathcal{P}_{\ell}}u_P
    \right) \\
    &\leq q^{\ell-2}\bigl(sN+(q-1)b_s\bigr)
     =A_s.
\end{align*}
\end{proof}

\subsection{Evaluation in a function field tower with many degree-$d$ places}

To instantiate the preceding construction, we need an explicit
family of function fields with many degree-$d$ places, together with efficient
algorithms for computing Riemann--Roch spaces and evaluating their elements at
these places. We obtain these ingredients by descending the
Garcia--Stichtenoth tower over $\FF_{q^d}$ to $\FF_q$. We call the resulting
family the $d$-descended Garcia--Stichtenoth function fields.

\begin{proposition}[Function field family with many degree-$d$ places]
\label{cor:degree-d-places-gs-tower}
Let $q$ be a prime power, let $d\geq2$ be even, and set $r=q^{d/2}$.
Let $(F_i)_{i\geq1}$ be the tower over $\FF_q$ defined by
\[
    F_1=\FF_q(x_1),
    \qquad
    F_{i+1}=F_i(x_{i+1}),
\]
where
\[
    x_{i+1}^{r}+x_{i+1}
    =\frac{x_i^{r}}{x_i^{r-1}+1}.
\]
Let $g_i=g(F_i)$. There exists an explicit set $\mathcal{P}_d^{(i)}$ of
degree-$d$ places in $F_i$, represented by evaluation tuples in
$\FF_{q^d}^i$, such that
\[
    B_d(F_i) \geq |\mathcal{P}_d^{(i)}|\geq\frac{q^{d/2}-1}{2d}\,g_i.
\]
For every $i\geq1$, there exists a unique pole $P_\infty$ of $x_1$ in
$F_i$ and it is rational. For every integer $\mu\geq0$, an explicit $\FF_q$-basis of
$\mathcal{L}(\mu P_\infty)$ can be computed in time
$\poly(q,d,g_i,\mu)$. Moreover, given a degree-$d$ place represented by an evaluation tuple in $\FF_{q^d}^i$, every element of $\mathcal{L}(\mu P_\infty)$ can be evaluated at this place in time
$\poly(q,d,g_i,\mu)$.
\end{proposition}

The proof of \cref{cor:degree-d-places-gs-tower} is given in
\cref{app:function-fields-missing-proofs}.

\begin{theorem}[Strong subspace designs over arbitrary finite fields]
\label{thm:arbitrary-field-strong-subspace-designs}
Let $q$ be a prime power, let $\ell\geq2$ be even, and let $k\geq2$.
For every $1\leq r\leq k$, there exists an explicit collection
$\mathcal{H}$ of $n$ subspaces of $\FF_q^k$, each of codimension at most
$r$, such that, for every $1\leq s\leq r$, the collection is an
$(s,A_s)$-strong subspace design, where
\[
    A_s=
      C_\ell(q)sn, \qquad C_\ell(q)=\frac{1}{q}+\frac{12(q-1)}{q(q^{\ell/2}-1)}.
\]
Moreover,
\[
    n\leq q^{\ell-1}\left\lceil
      \frac{q^{\ell/2}-1}{2\ell}q^\ell k
    \right\rceil.
\]
\end{theorem}

\begin{proof}
Let $\widetilde{k}$ be the smallest genus in the $\ell$-descended
Garcia--Stichtenoth tower that is at least $k$, and let
$F/\FF_q$ be the corresponding function field. The genus formula in
\cref{thm:gs-tower-parameters} shows that consecutive positive genera
grow by a factor of at most $q^{\ell/2}+1\leq q^\ell$. Hence
\[
    k\leq\widetilde{k}<q^\ell k.
\]
Let $P_\infty$ be the unique pole of $x_1$ in $F$, and set
$
    d=2\widetilde{k}-1.
$
By \cref{fact:Bertrand}, there is a prime number $m$ such that $d<m<2d$, and we can
find such a prime deterministically by enumeration. Applying
\cref{prop:explicit-prime-degree-place} to the $\ell$-descended tower gives
an explicit place $Q$ of degree $m$.
Set
\[
    N=\left\lceil
      \frac{q^{\ell/2}-1}{2\ell}\widetilde{k}
    \right\rceil.
\]
By \cref{cor:degree-d-places-gs-tower}, we may choose a set
$\mathcal{P}_{\ell}$ of $N$ degree-$\ell$ places of $F$. Applying
\cref{prop:affine-evaluation-agpv} gives an explicit collection
$\widetilde{\mathcal H} = \mathcal{H}(F, \mathcal{P}_\ell, P_\infty, Q, \widetilde{k}, r, \ell)$ of $n=q^{\ell-1}N$ subspaces in $\FF_q^{\widetilde{k}}$, each of codimension at most $r$,
which form an $(s, A_s)$-subspace design for every $1 \leq s \leq r$.

Fix some $k$-dimensional subspace $U\leq V$, identify it with
$\FF_q^k$, and set
\[
    \mathcal H:=(H\cap U)_{H\in\widetilde{\mathcal H}}.
\]
For every $H\in\widetilde{\mathcal H}$ and $W\leq U$, we have
\[
    W\cap(H\cap U)=W\cap H,\qquad
    \operatorname{codim}_U(H\cap U)\leq\operatorname{codim}_V H\leq r.
\]
Thus, $\mathcal{H}$ is a collection of $n$ subspaces in $\FF_q^k$, each of codimension at most $r$,
which forms an $(s,A_s)$-subspace design.
Moreover, since $m<2d$, we have
$d+(s-1)\widetilde d<6s\widetilde{k}$. Thus, for every $s$-dimensional
subspace $W\leq U$ we obtain,
\[
    \frac{1}{n}\sum_{H\in\mathcal H}\dim_{\FF_q}(W\cap H)
    <\left(
      \frac{1}{q}+\frac{6(q-1)\widetilde{k}}{q\ell N}
    \right)s
    \leq C_\ell(q)s.
\]
This proves the asserted value of $A_s$. The bound on $n$ follows from
$\widetilde{k}<q^\ell k$.

Finally, the construction is explicit by
\cref{prop:explicit-prime-degree-place,cor:degree-d-places-gs-tower}.
The former gives the place $Q$, an explicit representation of its residue
field, and the residue map at $Q$. The latter gives bases of the required
Riemann--Roch spaces on the descended tower, as well as the degree-$\ell$
places and the evaluations at these places. Using these, we compute
the right inverse $I$, the lifted Frobenius maps $\tau_j$, and the kernels
defining the subspaces in $\widetilde{\mathcal H}$ by linear algebra over
$\FF_q$. The subspace $U$ and the intersections defining $\mathcal H$ are
obtained from $\widetilde{\mathcal{H}}$ by linear algebra.
\end{proof}

The following corollary instantiates the preceding theorem so that the
normalized intersection bound is within an arbitrary additive $\gamma$ of
$1/q$.

\begin{corollary}[Strong subspace designs over arbitrary finite fields]
\label{cor:all-field-strong-subspace-designs}
Let $q$ be a prime power, let $0 < \gamma < 1$, and let $k\geq2$. For every
$1\leq r\leq k$,
there exists an explicit collection $\mathcal{H}$ of $n$ subspaces of
$\FF_q^k$, each of codimension at most $r$, satisfying
\[
    n\leq q\left(1+\frac{12}{\gamma}\right)^2
      \left\lceil
        \frac{3q^2(1+12/\gamma)^2k}{\gamma}
      \right\rceil
    =O\left(q^3k/\gamma^{5}\right).
\]
For every $1\leq s\leq r$, the collection is an
$(s,A_s)$-strong subspace design, where
\[
    A_s=\left(\frac{1}{q}+\gamma\right)sn.
\]
In particular, every nonzero subspace $W\leq\FF_q^k$ of dimension at most
$r$ satisfies
\[
    \frac{1}{n}\sum_{H\in\mathcal{H}}
    \dim_{\FF_q}(W\cap H)
    <\left(\frac{1}{q}+\gamma\right)\dim_{\FF_q}W.
\]
\end{corollary}
\begin{proof}
Set
\[
    \ell=2\left\lceil
      \log_q\left(1+\frac{12}{\gamma}\right)
    \right\rceil.
\]
Let $\widetilde{k}$ be the tower genus chosen in the proof of
\cref{thm:arbitrary-field-strong-subspace-designs}, and set
\[
    N=\left\lceil\frac{6\widetilde{k}}{\ell\gamma}\right\rceil.
\]
Then $\ell$ is even and
\[
    q^{\ell/2}-1\geq\frac{12}{\gamma},
\]
and hence $C_\ell(q)-1/q<\gamma$. Moreover,
\[
    N\leq\left\lceil
      \frac{q^{\ell/2}-1}{2\ell}\widetilde{k}
    \right\rceil.
\]
Applying the construction in the proof of
\cref{thm:arbitrary-field-strong-subspace-designs} with this choice of
$N$ gives a collection of size $n=q^{\ell-1}N$. If
$\dim_{\FF_q}W=s\leq r$, then
\begin{align*}
    \frac{1}{n}\sum_{H\in\mathcal{H}}
      \dim_{\FF_q}(W\cap H)
    &<\left(
      \frac{1}{q}+\frac{6(q-1)\widetilde{k}}{q\ell N}
    \right)s \\
    &\leq\left(
      \frac{1}{q}+\frac{q-1}{q}\gamma
    \right)s
    <\left(\frac{1}{q}+\gamma\right)s.
\end{align*}
This proves the asserted value of $A_s$.

Finally,
\[
    q^{\ell-1}\leq q\left(1+\frac{12}{\gamma}\right)^2,
    \qquad
    N\leq\left\lceil
      \frac{3q^2(1+12/\gamma)^2k}{\gamma}
    \right\rceil,
\]
where the second inequality follows from
$\widetilde{k}<q^\ell k\leq q^2(1+12/\gamma)^2k$. This gives the claimed
bound on $n$.
\end{proof}

\subsection{Rank condensers and extractors}
\label{subsec:rank-condensers-extractors}

We now record the consequences for lossy rank condensers and lossless rank extractors.

\begin{corollary}[Lossy rank condensers over arbitrary finite fields]
\label{cor:all-field-lossy-rank-condensers}
Let $\varepsilon,\delta\in(0,1]$, and let $q$ be a prime power satisfying
\[
    q>\frac{1}{\varepsilon\delta}.
\]
For every integer $k\geq2$ and every $1\leq r\leq k$, there exists an explicit
$(r,\varepsilon,\delta)$-lossy rank condenser
$\mathcal{E}=(E_i)_{i\in[n]}\subseteq\FF_q^{r\times k}$ of size
\[
    n=O\left(
      \frac{q^3k}
      {\bigl(\varepsilon\delta-1/q\bigr)^5}
    \right).
\]
\end{corollary}
\begin{proof}
Set
\[
    \gamma=\varepsilon\delta-\frac{1}{q},
    \qquad
    \ell=2\left\lceil
      \log_q\left(1+\frac{12}{\gamma}\right)
    \right\rceil.
\]
Let $\widetilde{k}$ be the tower genus chosen in the proof of
\cref{thm:arbitrary-field-strong-subspace-designs}, and set
\[
    N=\left\lceil\frac{6\widetilde{k}}{\ell\gamma}\right\rceil.
\]
The choice of $\ell$ gives
\[
    q^{\ell/2}-1\geq\frac{12}{\gamma},
\]
and hence
\[
    N\leq
    \left\lceil
      \frac{q^{\ell/2}-1}{2\ell}\widetilde{k}
    \right\rceil.
\]
We may therefore apply the construction in the proof of
\cref{thm:arbitrary-field-strong-subspace-designs} with this choice of
$N$. By \cref{prop:affine-evaluation-agpv}, the resulting collection is
an $(r,A_r)$-strong subspace design of size $n=q^{\ell-1}N$, where
\[
    \frac{A_r}{rn}
    <\frac{1}{q}+\frac{6(q-1)\widetilde{k}}{q\ell N}
    \leq\frac{1}{q}+\frac{q-1}{q}\gamma
    <\varepsilon\delta.
\]
Consequently,
\[
    A_r<\varepsilon\delta rn
    <\bigl(\lfloor\varepsilon r\rfloor+1\bigr)
      \bigl(\lfloor\delta n\rfloor+1\bigr),
\]
and the claim follows from
\cref{prop:strong-design-to-lossy-condenser}.

Finally, the definition of $\ell$ gives
$q^{\ell-1}=O(q/\gamma^2)$ and $q^\ell=O(q^2/\gamma^2)$, while
$N=O(\widetilde{k}/\gamma)=O(q^2k/\gamma^3)$. This proves the stated
bound on $n$, since $\gamma=\varepsilon\delta-1/q$.
\end{proof}

Taking the rank loss in the preceding corollary to be smaller than $1/r$
gives lossless rank extractors.

\begin{corollary}[Explicit lossless rank extractors]
\label{cor:all-field-lossless-rank-extractors}
Let $1 \leq r \leq k$ be integers, and let $0 < \delta < 1$. For every prime power
$q>r/\delta$, there exists an explicit
$(r,\delta)$-lossless rank extractor $\mathcal{E}=(E_i)_{i\in[n]}\subseteq\FF_q^{r\times k}$ of size
\[
    n=O\left(
      \frac{q^3r^5k}{(\delta-r/q)^5}
    \right).
\]
In particular, for all integers $1\leq r \leq k$ and every prime power $q>r$, there exists an explicit $r$-lossless rank disperser over
$\FF_q$ of size $O(r^5kq^8)$.
\end{corollary}
\begin{proof}
Let
\[
    \varepsilon
    =\frac{1}{2}\left(\frac{1}{r}+\frac{1}{q\delta}\right).
\]
Since $q>r/\delta$, we have $\varepsilon r<1$. Moreover,
\[
    \varepsilon\delta-\frac{1}{q}
    =\frac{1}{2}\left(\frac{\delta}{r}-\frac{1}{q}\right)
    =\frac{\delta-r/q}{2r}>0.
\]
Applying \cref{cor:all-field-lossy-rank-condensers} with loss parameter
$\varepsilon$ and bad-map fraction $\delta$ gives an explicit collection of
size
\[
    n=O\left(
      \frac{q^3r^5k}{(\delta-r/q)^5}
    \right).
\]
Every good map preserves rank at least
$(1-\varepsilon)r>r-1$, and hence preserves rank $r$. Thus, at most
$\lfloor\delta n\rfloor\leq\delta n$ maps fail to preserve the
rank of each rank-$r$ input.

For the final assertion, assume that $r<q$ and set
$\delta_0=(2r+1)/(2q)$. Then $r/q<\delta_0<1$, and the preceding
bound becomes $n=O(r^5kq^8)$. Since $\delta_0<1$, at least one map
preserves the rank of each rank-$r$ input, so the collection is an
$r$-lossless rank disperser.
\end{proof}

\subsection{Strong blocking sets over arbitrary finite fields}
\label{subsec:strong-blocking-sets}

Fancsali and Sziklai \cite{FS14lines, FS16lines} gave a reduction from
lossless rank dispersers to strong blocking sets
(stated formally in \cite[Lemma~6.1]{GRSZ26}). Hence, our results imply new strong blocking sets.
We begin with the definitions.

\begin{definition}[Projective space and projective subspace]
The projective space $\operatorname{PG}(k-1,q)$ is the set of
one-dimensional subspaces of $\FF_q^k$. For every nonzero subspace
$U\leq\FF_q^k$, the set
\[
    \operatorname{PG}(U)
    =\{L\leq U:\dim_{\FF_q}(L)=1\}
\]
is a \emph{projective subspace} of $\operatorname{PG}(k-1,q)$. Its dimension
is $\dim_{\FF_q}(U)-1$, and its codimension in
$\operatorname{PG}(k-1,q)$ is $k-\dim_{\FF_q}(U)$. For a nonempty set
$S\subseteq\operatorname{PG}(k-1,q)$, its \emph{projective span} is the
smallest projective subspace containing $S$.
\end{definition}

\begin{definition}[Strong blocking set]
Let $1\leq s\leq k-1$. A set
$B\subseteq\operatorname{PG}(k-1,q)$ is a \emph{strong $s$-blocking set} if,
for every codimension-$s$ projective subspace
$\Sigma\subseteq\operatorname{PG}(k-1,q)$, the intersection
$B\cap\Sigma$ spans $\Sigma$; that is, $\Sigma$ is the smallest projective
subspace containing $B\cap\Sigma$.
\end{definition}
We next describe the above reduction.
Let $\mathcal{E}=(E_i)_{i\in[M]}\subseteq\FF_q^{(s+1)\times k}$ be an
$(s+1)$-lossless rank disperser, and set
$V_i=\operatorname{rowspan}(E_i)$. Then
\[
    B=\bigcup_{i\in[M]}\operatorname{PG}(V_i)
\]
is a strong $s$-blocking set in $\operatorname{PG}(k-1,q)$ and
\begin{equation}
\label{eq:rank-disperser-to-blocking-set}
    |B|\leq M\frac{q^{s+1}-1}{q-1}
    \leq2Mq^s.
\end{equation}
Using our lossless rank extractors, we obtain the following corollary.

\begin{corollary}[Strong blocking sets over arbitrary finite fields]
\label{cor:prime-field-strong-blocking-sets}
Let $q$ be a prime power, and let $s\geq1$ satisfy
\[
    q>s+1.
\]
For every integer $k\geq s+1$, there exists an explicit strong $s$-blocking set
$B\subseteq\operatorname{PG}(k-1,q)$ satisfying
\[
    |B|=O\bigl(s^5kq^{s+8}\bigr).
\]
\end{corollary}
\begin{proof}
Set $r=s+1$. Since $q>r$,
\cref{cor:all-field-lossless-rank-extractors} gives an explicit
lossless rank disperser of size $M=O(r^5kq^8)$. Thus,
\cref{eq:rank-disperser-to-blocking-set} gives an explicit strong
$s$-blocking set satisfying
\[
    |B|\leq2Mq^s
    =O\bigl(r^5kq^{s+8}\bigr)
    =O\bigl(s^5kq^{s+8}\bigr).
\]
\end{proof}

\paragraph{Comparison with previous results.}
\Cref{tab:strong-blocking-set-comparison} compares our construction with
the results in \cite[Table~1]{GRSZ26}. Let $c_0>0$ be the sufficiently large absolute constant in
\cite[Corollary~6.2]{GRSZ26}. We first consider the
range $s+1<q<(2s)^{c_0}$. In this range, our bound
$O(s^5kq^{s+8})=kq^s s^{O(1)}$ improves all applicable previous bounds
in the table as $s\to\infty$. In particular, it replaces the exponential
factors appearing in \cite[Corollary~6.2(3), Proposition 6.8]{GRSZ26} and \cite[Theorem 17]{BT26} by a polynomial in $s$.
We next consider $q\geq(2s)^{c_0}$. For prime $q$, our
construction replaces the quasipolynomial factor
$(2s)^{O(\log(2s))}$ in \cite[Corollary~6.2(2)]{GRSZ26} by
$O(s^5q^8)$, which is polynomial in $s$ and $q$. Thus, whenever
$q=s^{O(1)}$, it improves the dependence on $s$ from quasipolynomial
to polynomial. For non-prime $q$ in this upper range,
\cite[Corollary~6.2(1)]{GRSZ26} retains the smaller bound
$O(s(k-s)q^{s+1})$ for every $k\geq s+1$, and the bound
$O(s(k-s)q^s)$ for infinitely many such $k$.
Finally, when $q\leq s+1$, our result does not apply, whereas
\cite[Corollary~6.2(3), Proposition~6.8]{GRSZ26} and
\cite[Theorem~17]{BT26} apply over every prime power.

\begin{table}[htbp]
\centering
\setlength{\belowcaptionskip}{6pt}
\caption{Explicit strong $s$-blocking sets in
$\operatorname{PG}(k-1,q)$, extending \cite[Table~1]{GRSZ26}.
The bounds shown hold for every $k\geq s+1$; the bounds in
\cite[Corollary~6.2]{GRSZ26} improve by a factor of $q$ for infinitely
many $k$.
Here $C>0$ is a sufficiently large absolute constant, $o_{qs}(1)$ tends to zero as $qs\to\infty$, and
$\mu_{s,q}=\max\{C\log(2s)/\log q,2\}$.}
\label{tab:strong-blocking-set-comparison}
\begingroup
\small
\setlength{\tabcolsep}{4pt}
\renewcommand{\arraystretch}{1.35}
\begin{tabular}{@{}>{\raggedright\arraybackslash}p{0.33\textwidth}
                  >{\raggedright\arraybackslash}p{0.27\textwidth}
                  >{\raggedright\arraybackslash}p{0.36\textwidth}@{}}
\toprule
Reference & Field size & Upper bound on $|B|$ \\
\midrule
\cite[Theorem~16]{BT26}
  & $q$ square, $q\geq(2s)^{Cs}$
  & $2^{O(s^2\log(2s))}kq^s$ \\
\cite[Theorem~17]{BT26}
  & Every prime power
  & $kq^{O(s^2)}$ \\
\cite[Corollary~6.2(1)]{GRSZ26}
  & $q$ non-prime, $q\geq(2s)^{C}$
  & $O(s(k-s)q^{s+1})$ \\
\cite[Corollary~6.2(2)]{GRSZ26}
  & $q$ prime, $q\geq(2s)^{C}$
  & $(2s)^{O(\log(2s))}(k-s)q^{s+1}$ \\
\cite[Corollary~6.2(3)]{GRSZ26}
  & Every prime power
  & $O(\mu_{s,q}^{s+1}s(k-s)q^{s+1})$ \\
\cite[Proposition~6.8]{GRSZ26}
  & Every prime power
  & $kq^{(2+o_{qs}(1))s}$ \\
\midrule
\Cref{cor:prime-field-strong-blocking-sets} 
  & $q>s+1$
  & $O(s^5kq^{s+8})$ \\
\bottomrule
\end{tabular}
\par\smallskip
\endgroup
\end{table}

\subsection{Condensers for affine sources}
\label{subsec:affine-source-condensers}

An affine $(n,k)$-source over $\FF_q$ is the uniform distribution on an
affine subspace of $\FF_q^n$ of dimension at least $k$. A seeded condenser for affine sources
maps such a source to a shorter one while retaining most of its
entropy. For a discrete random variable $Z$, its \emph{min-entropy}%
\footnote{For an affine source, min-entropy equals Shannon entropy because
the distribution is uniform on its support.}
is
\[
    H_\infty(Z)=-\log \left(\max_z\Pr[Z=z]\right).
\]
Seeded condensers for \emph{general} weak sources, introduced in \cite{RR99,RSW06}, are a standard tool in pseudorandomness. Here we
restrict attention to affine sources. Although \emph{deterministic} condensers exist
for affine sources, the seeded setting permits constructions that are linear
in the source for every fixed seed, a property used in several works.

\begin{definition}[seeded condenser for affine sources \cite{GR08a,DG26}]
\label{def:seeded-affine-condenser}
Let $q$ be a prime power, let $1\leq k\leq n$ and $m,D\geq1$ be
integers, and let
$0\leq k'\leq\min\{k,m\}$ and $\varepsilon\in[0,1]$.
We say that a map
\[
    \operatorname{Cond}\colon\FF_q^n\times[D]\longrightarrow\FF_q^m
\]
is a \emph{seeded affine $(k\to_{\varepsilon}k')$-condenser} if,
for every affine $(n,k)$-source $X$,
\[
    \bigl|\{y\in[D]:
      H_\infty(\operatorname{Cond}(X,y))<k'\log q\}\bigr|
    \leq\varepsilon D.
\]
Its seed length is $d=\log D$. The condenser is \emph{lossless} if
$k'=k$.
The condenser is \emph{linear} if for every fixed seed
$y \in [D]$, the map $\text{Cond}(\cdot, y):\FF_q^n \longrightarrow \FF_q^m$ is linear.\footnote{An alternative definition requires that, for every affine
$(n,k)$-source $X$ and an independent uniform seed $Y\in[D]$, the joint
distribution $(Y,\operatorname{Cond}(X,Y))$ be $\varepsilon$-close to a
distribution with min-entropy at least $d+k'\log q$. The definition above
implies this condition.
Conversely, for linear condensers, the alternative definition with error
$\varepsilon$ implies the definition above with error at most
$\varepsilon/(1-1/q)\leq2\varepsilon$ (\cite[Proposition 2.9]{Rao08}). Thus, for linear condensers, the two definitions are equivalent up to a factor of $2$ in
the error.}
\end{definition}

For $i\in[D]$, let $E_i\colon\FF_q^n\to\FF_q^m$ be a linear map with
$m\geq k$, and set $\operatorname{Cond}(x,i)=E_ix$. If $X$ is uniform
on $x_0+W$, then
\[
    H_\infty(E_iX)=\operatorname{rank}(E_i|_W)\log q.
\]
Thus, a $(k,\varepsilon,\delta)$-lossy rank condenser is a linear seeded affine
$(k\to_{\delta}(1-\varepsilon)k)$-condenser. This equivalence was recently 
made explicit, and refined, in \cite{DG26}. Our work thus readily implies the following.

\begin{corollary}[Condensers for affine sources]
\label{cor:affine-source-condensers}
Let $q$ be a prime power and let $\delta,\varepsilon\in(0,1)$ satisfy
$\delta\varepsilon>1/q$.
For every integer $n\geq1$ and every $1\leq k\leq n$, there exists an
explicit linear seeded affine
$(k\to_{\varepsilon}(1-\delta)k)$-condenser
$
    \operatorname{Cond}\colon\FF_q^n\times[D]\longrightarrow\FF_q^k.
$
Its seed length satisfies
\[
    d\leq\log n+3\log q
      +5\log \left(\frac{1}{\delta\varepsilon-1/q}\right)+O(1).
\]
For fixed $q,\delta,$ and $\varepsilon$, the construction has
$d=\log n+O_{q,\delta,\varepsilon}(1)$.
\end{corollary}
\begin{proof}
Set $\gamma=\delta\varepsilon-1/q$.
Apply \cref{cor:all-field-lossy-rank-condensers} with ambient dimension
$n$, rank $k$, loss parameter $\delta$, and bad-map fraction
$\varepsilon$, and set $\operatorname{Cond}(x,i)=E_ix$ for the resulting
maps $E_i\in\FF_q^{k\times n}$. The
construction in that corollary uses
\[
    D=O\left(\frac{q^3n}{(\delta \varepsilon - 1/q)^5}\right).
\]
Taking logarithms gives the seed-length bound.
\end{proof}

Note that with, say, $\eps=\delta=3/4$, one can even take $q=2$.

\paragraph{Comparison with previous results}
\Cref{tab:affine-condenser-comparison} compares the seed length, output
dimension and parameter restrictions of previous
constructions. For every fixed $q,\delta,$ and $\varepsilon$ satisfying
$\delta\varepsilon>1/q$, our construction has $m=k$ and seed length
$\log n+O(1)$.
Over constant-sized fields, the earlier
bounds below either get worse seed length or vanishing output entropy rate. Our construction is the first to achieve both
optimal seed length dependence on $n$,
and constant output entropy rate, over all fixed fields -- and in particular over $\FF_2$ (when $\eps,\delta$ are large enough).

The lossless rank extractors of Gabizon and Raz
\cite[Theorem~5]{GR08a} and Forbes and Shpilka
\cite[Theorem~4.1]{FS12} give lossless affine extractors.
Forbes and Shpilka improve the number of bad maps from $nk^2$ to
$nk-\binom{k+1}{2}$, thereby saving a $\log k$ term in the seed length.
However, both families are indexed by elements of $\FF_q$, so achieving
error $\varepsilon$ requires $q=\Omega(nk^2/\varepsilon)$ and
$q=\Omega(nk/\varepsilon)$, respectively.

The lossless condensers of Guruswami, Umans, and Vadhan \cite[Theorem~1.7]{GUV09},
later made linear by Cheraghchi \cite[Corollary 2.23]{Che10},
apply to arbitrary sources over any field, at the cost of a constant-factor increase in seed length.
Their second construction~\cite[Theorem~4.4]{GUV09} avoids this
constant-factor increase in seed length, but has vanishing output entropy rate.

The folded Reed-Solomon construction of Guruswami and Kopparty
\cite[Theorem~14]{GK16}, together with the lossy analysis of Forbes and
Guruswami \cite[Proposition~6.8]{FG15}, allows every $k\leq m<n$.
Its seed length is
$\log (n/(m-k+1))+\log (1/(\delta\varepsilon))+O(1)$, but requires
the field size be larger than $n$.
Applying the field reduction of \cite[Proposition~8.5]{FG15} gives seed length
$\log (n/k)+\log (1/(\delta\varepsilon))+O(1)$ over every field, at
the cost of vanishing output entropy rate.

Guruswami, Xing, and Yuan \cite[Theorem~1.2]{GXY18} gave a function field generalization
of the construction of \cite{GK16}, 
which relaxes the
field-size restriction,
but their condenser bounds require either a growing field
or a vanishing output entropy rate. Indeed, for output dimension
$m\geq4k$, their construction supplies
$\Omega((n/m)q^{\lfloor m/(2k)\rfloor})$ maps, whereas $O(n \log_q n /(m\delta\varepsilon))$ maps suffice to obtain a $(k\to_{\varepsilon}(1-\delta)k)$-condenser.
Thus their bounds guarantee sufficiently many maps when
$
    q^{\lfloor m/(2k)\rfloor}
    \gtrsim\frac{\log_q n}{\delta\varepsilon}.
$
For fixed loss and error, keeping $m=O(k)$ therefore requires $q$
to grow with $n$, although polylogarithmic field size suffices.
Over a fixed field, this condition instead requires
$m/k=\Omega(\log_q\log_q n)$, so the guaranteed output entropy rate
$(1-\delta)k/m$ tends to zero.

Goyal, Guruswami, and Hsieh \cite{GGH26} work over every field size and allow arbitrary positive loss and error, but their
output dimension grows exponentially with $k^2$, so their guaranteed
output entropy rate vanishes as $k$ grows.

Finally,
Guo, Raj, Shangguan, and Zhang \cite{GRSZ26} obtain $m=k$ and seed
length $\log n+O(1)$ for infinitely many input dimensions, but require $q$ to grow polynomially with $k$;
over prime fields, their seed length also contains an additional
polylogarithmic term in $k$.

In the $q=2$ regime, a recent work of Doron and Goodman \cite{DG26} interpreted subspace designs as linear seeded condensers for affine sources, and combined \cite{GK16,FG15} with other objects from extractor theory. They were able to get a $(k \rightarrow_{\eps} (1/2+\beta)m)$ seeded linear affine condenser from $\mathbb{F}_2^n$ to $\mathbb{F}_2^m$ with seed length $\log(n/k)+O(\log\log n)$ and $m \ge \frac{k}{\operatorname{polylog(n)}}$ for \emph{some} nontrivial $\eps$, and some constant $\beta > 0$. For an arbitrary $\eps > 0$, they get seed length $O(\log(n/k)+\log(1/\eps)+\log\log n)$.

\begin{table}[p]
\centering
\setlength{\belowcaptionskip}{6pt}
\caption{$(k\to_{\varepsilon}(1-\delta)k)$ seeded linear affine condensers from $\FF_q^n$ to $\FF_q^m$ with seed length $d$.
Seed lengths are upper bounds, with additive $O(1)$ terms omitted.
All implicit constants are absolute.
When $\delta$ does not appear, the construction is lossless (i.e. $\delta=0$).
The output entropy rate is $\frac{(1-\delta)k}{m}$.}
\label{tab:affine-condenser-comparison}
\begingroup
\small
\setlength{\tabcolsep}{4pt}
\renewcommand{\arraystretch}{1.2}
\begin{tabular}{@{}>{\raggedright\arraybackslash}p{0.22\textwidth}
                  >{\raggedright\arraybackslash}p{0.20\textwidth}
                  >{\raggedright\arraybackslash}p{0.23\textwidth}
                  >{\raggedright\arraybackslash}p{0.29\textwidth}@{}}
\toprule
Reference & Field size & Output dimension $m$ & Seed length $d$ \\
\midrule
\cite[Theorem~5 and Remark~6.2]{GR08a}
 & $q \geq \Omega(nk^2/\varepsilon)$
 & $k$
 & $\log (n/\varepsilon)+2\log k$\\
\midrule
\cite[Theorem~4.1]{FS12}
 & $q\geq \Omega(nk/\varepsilon)$
 & $k$
 & $\log (nk/\varepsilon)$ \\
\midrule
\cite[Corollary 2.23]{Che10}
 & Every $q$
 & $O(k)+d$
 & $O(\log (nk/\varepsilon) + \log(q))$ \\
\midrule
\cite[Theorem~4.4]{GUV09}
 & Every $q$
 & $O(kd)$
 & $\log (nk/\varepsilon) + \log(q)$ \\
\midrule
\cite[Theorem~14]{GK16}
 & $q\geq \Omega \left( \frac{mn}{\delta\varepsilon(m-k+1)} \right )$
 & $m\geq k$
 & $\log\!\dfrac{n}{m-k+1}$ $+\log\!\dfrac{1}{\delta\varepsilon}$ \\
\midrule
\cite[Proposition~8.5]{FG15}
 & Every $q$
 & $O\!\left(k\left(1+\log_q\!\dfrac{n}{\delta\varepsilon}\right)\right)$
 & $\log(n/(\delta\varepsilon k))$ \\
\midrule
\cite[Theorem~1.2]{GXY18}
 & $q\geq\left(\dfrac{\log n}{\delta\varepsilon}\right)^{O(k/m)}$
 & $m\geq4k$
 & $\log(n/(\delta\varepsilon m))$ $+\log\log_q n$ \\
\midrule
\cite[Theorem~1.1]{GGH26}
 & Every $q$
 & $\poly(k,1/(\delta\varepsilon))q^{k^2}$
 & $\log(n/(m\delta\varepsilon))$ \\
\midrule
\cite[Theorem~4.5]{GRSZ26}
 & $q\geq\poly(k,1/(\delta\varepsilon))$; non-prime
 & $k$
 & $\log (n/(\delta\varepsilon))$ \\
\midrule
\cite[Theorem~5.10]{GRSZ26}
 & $q\geq\poly(k,1/(\delta\varepsilon))$; prime
 & $k$
 & $\log (n/(\delta\varepsilon))$
   $+O\left(\log k \cdot\log(k/(\delta\varepsilon))\right)$ \\
\midrule
\cite{DG26}
 & $q=2$
 & $\textnormal{see discussion above}$
 &  \\
\midrule
\Cref{cor:affine-source-condensers}
 & $q>1/(\delta\varepsilon)$
 & $k$
 & $\log n+O\!\left(\log\!\dfrac{q}{\delta\varepsilon-1/q}\right)$ \\
\bottomrule
\end{tabular}
\par\smallskip
\begin{minipage}{\textwidth}
\footnotesize
\end{minipage}
\endgroup
\end{table}

\paragraph{AI Usage.} All ideas, constructions, and proofs, were developed by the authors, except for the case $\text{char}\,\FF_q \neq 2$ of \cref{prop:explicit-prime-degree-place}, for which we assisted GPT-5.6 Sol. GPT-5.6 Sol was also used for language editing.

\FloatBarrier

\bibliographystyle{alpha}
\bibliography{bib}

\clearpage
\appendix
\crefalias{section}{appendix}
\section{Missing Proofs}
\label{app:function-fields-missing-proofs}

We begin with the proof of the case $\text{char}\,\FF_q = 2$ of \cref{prop:explicit-prime-degree-place}.
\begin{proof}[Proof of \cref{prop:explicit-prime-degree-place}]

It remains to consider characteristic 2. Write $q=2^e$ and $r=2^s$, so that $e\ell=2s$. Using Shoup's algorithm, construct an
irreducible polynomial $h\in\FF_2[T]$ of degree $m$, and set
\[
    L=\FF_2[T]/(h),\qquad K=\FF_q[T]/(h).
\]
The condition $\gcd(m,e)=1$ implies that $h$ remains irreducible over
$\FF_q$. Thus $L\cong\FF_{2^m}$ is a subfield of
$K\cong\FF_{q^m}$. Write $\tau=\Tr_{L/\FF_2}$; since $m$ is odd,
$\tau(1)=1$. The $\FF_2$-linear map
\[
    A\colon L\longrightarrow L,\qquad A(z)=z^r+z
\]
has kernel $\FF_2$. Indeed, $A(z)=0$ if and only if $z^{2^s}=z$, and the
fixed field of $z\mapsto z^{2^s}$ in $L\cong\FF_{2^m}$ is
$\FF_{2^{\gcd(s,m)}}=\FF_2$. Moreover, Frobenius invariance of the trace gives
\[
    \tau(A(z))=\tau(z^{2^s})+\tau(z)=0
\]
for every $z\in L$, and hence $\im(A)\subseteq\ker(\tau)$. Both spaces have
dimension $m-1$ over $\FF_2$: the first because $\dim_{\FF_2}\ker(A)=1$, and
the second because $\im(\tau) = \FF_2$. Therefore
$\im(A)=\ker(\tau)$.
Consequently, whenever $\tau(b)=0$, we can solve $A(z)=b$ by
Gaussian elimination over $\FF_2$.

We first construct $c\in L^\times$ such that
\[
    \FF_2(c)=L,\qquad \tau(c)=\tau(c^{-1})=0.
\]
Set $M=2m+1$. Since $m\geq5$, the group $L^\times$ contains more than $M$
elements. Enumerate any $M+1$ of them, and choose one, denoted by $u$, such
that $u^M\neq1$. Such an element exists
because $X^M-1$ has at most $M$ roots. In particular, $u\notin\FF_2$, so $u$
has degree $m$ over $\FF_2$, since the prime-degree extension $L/\FF_2$ has no
proper intermediate fields. Its Frobenius orbit
is disjoint from that of $u^{-1}$. Indeed, an equality
$u^{2^a}=u^{-1}$ with $0\leq a<m$ would, after applying the $2^a$-power
Frobenius once more, give
\[
    u^{2^{2a}}=(u^{-1})^{2^a}=(u^{2^a})^{-1}=u.
\]
Since the Frobenius orbit of $u$ has size $m$, this gives $m\mid2a$, and hence
$a=0$. But then $u=u^{-1}$, which in characteristic two implies $u=1$, a
contradiction. Thus
\[
S=\{1\}\cup\{u^{2^a}:0\leq a<m\}
      \cup\{u^{-2^a}:0\leq a<m\}
\]
has $M$ distinct nonzero elements.

We claim that there is some $1\leq j\leq M$ with $\tau(u^j+u^{-j})=0$.
Otherwise, for every such $j$ we
would have
\[
    \sum_{a\in S}a^j=1+\tau(u^j+u^{-j})=0.
\]
Enumerate $S=\{s_1,\ldots,s_M\}$ and consider the matrix
\[
    V=(s_k^j)_{1\leq j,k\leq M}.
\]
This matrix is invertible. Indeed, multiplying its $k$th column by $s_k^{-1}$ gives
the Vandermonde matrix $(s_k^{j-1})_{1\leq j,k\leq M}$, which is invertible
because the elements of $S$ are distinct. On the other hand, the preceding
equalities say that $V(1,\ldots,1)^\top=0$, a contradiction. We find a
suitable $j$ by testing all $M$ choices.

We next verify that $u^j\neq1$. Let $n$ be the multiplicative order of $u$.
For every positive integer $a$, we have
\[
    u^{2^a}=u \quad\Longleftrightarrow\quad n\mid 2^a-1.
\]
Since the Frobenius orbit of $u$ has size $m$, it follows that $m$ is the
multiplicative order of $2$ modulo $n$. Now choose a prime divisor $\rho$ of
$n$. As $n\mid 2^m-1$, the multiplicative order of $2$ modulo $\rho$ divides
$m$. It cannot equal $1$, since that would imply $\rho\mid 2-1$, and therefore
equals $m$ because $m$ is prime. By Fermat's little theorem this order divides
$\rho-1$, so $m\mid\rho-1$. Both $m$ and $\rho$ are odd, and hence
$2m\mid\rho-1$. Consequently,
\[
    n\geq\rho\geq2m+1=M.
\]
Equality $n=M$ would imply $u^M=1$, contrary to the choice of $u$. Thus
$n>M$.

Set $t=u^j$ and $c=t+t^{-1}$. Since $1\leq j\leq M<n$, we have $t\neq1$.
It follows that $c\neq0$, because in characteristic two the equality
$t+t^{-1}=0$ would imply $(t+1)^2=0$, and hence $t=1$. The choice of $j$
gives $\tau(c)=0$. We now verify the corresponding assertion for $c^{-1}$.
Since the characteristic is two,
\[
    c^{-1}
    =\frac{t}{t^2+1}
    =\frac{t}{(t+1)^2}
    =\frac{1}{t+1}+\frac{1}{(t+1)^2}.
\]
Let $w=(t+1)^{-1}$. Then $c^{-1}=w+w^2$, and Frobenius invariance gives
$\tau(w^2)=\tau(w)$. Therefore
\[
    \tau(c^{-1})=\tau(w)+\tau(w^2)=2\tau(w)=0.
\]
Finally, $c\notin\FF_2$, because $c\neq0$, $\tau(c)=0$, and $\tau(1)=1$.
Since $m$ is prime, $c$ therefore generates $L/\FF_2$.

We now construct the coordinates of $Q$. Solve $A(z)=c^{-1}$ and set
$a_1=z^{-1}$. Since $c^{-1}\neq0$ whereas $A$ vanishes on $\FF_2$, every
solution $z$ lies outside $\FF_2$. In particular, $z\neq0$, so $a_1$ is
well-defined, and
\[
    \phi(a_1)=\frac{1}{z^r+z}=c.
\]
Since $m$ is prime, $a_1\notin\FF_2$ generates $L/\FF_2$. Moreover,
$\tau(\phi(a_1))=0$. Suppose that $a_j\in L\setminus\FF_2$ has
been constructed with $\tau(\phi(a_j))=0$. Since the only elements
of $L$ fixed by $z\mapsto z^r$ lie in $\FF_2$, the value
$\phi(a_j)$ is defined and nonzero. Solve $A(b)=\phi(a_j)$. As
$\ker(A)=\FF_2$, the two solutions are $b$ and $b+1$; neither lies in
$\FF_2$, because their image under $A$ is nonzero. Furthermore,
\[
    \phi(b+1)+\phi(b)
    =\frac{(b+1)^{r+1}+b^{r+1}}{A(b)}
    =1+\frac{1}{\phi(a_j)}.
\]
Frobenius invariance gives
$\tau(a_j^{-r})=\tau(a_j^{-1})$, so
$1/\phi(a_j)=a_j^{-1}+a_j^{-r}$ has trace zero. Therefore the two values
$\phi(b)$ and $\phi(b+1)$ have different traces: taking the trace of both sides of the above equation gives
\[
    \tau(\phi(b+1))+\tau(\phi(b))=\tau(1)=1.
\]
Thus
\[
    a_{j+1}=b+\tau(\phi(b))
\]
selects $b$ if $\tau(\phi(b))=0$ and $b+1$ otherwise. It therefore satisfies
\[
    A(a_{j+1})=\phi(a_j),\qquad
    \tau(\phi(a_{j+1}))=0.
\]
In either case, $a_{j+1}\notin\FF_2$. This completes the induction and
constructs $a_1,\ldots,a_i$ without vanishing denominators.

The evaluation $x_j\mapsto a_j$ satisfies the defining equations of $F_i$, which defines a place $Q$.
Since $a_1$ generates $L/\FF_2$, its residue field is
\[
    \FF_q(a_1,\ldots,a_i)=\FF_q(a_1)=\FF_qL=K.
\]
Therefore $\deg Q=m$.
All steps can be carried out in time $\poly(q,\ell,g)$.
\end{proof}

Next, we prove \cref{cor:degree-d-places-gs-tower}.

\begin{proof}[Proof of \cref{cor:degree-d-places-gs-tower}]
The tuple $(0,\ldots,0)$ defines a rational place of $F_i$, so $F_i$ has
full constant field $\FF_q$. Its constant-field extension $F_i^{(d)}$ is
the $i$'th level of the Garcia--Stichtenoth tower over
$\FF_{q^d}=\FF_{r^2}$. It has the same genus $g_i$, and
$g_i\leq r^i$ by \cref{thm:gs-tower-parameters}.

Let $\mathcal{S}_i$ be the set of evaluation tuples
$(x_1,\ldots,x_i)\in\FF_{r^2}^i$ satisfying the tower equations and
$x_1^r+x_1\neq0$. The trace map
$z\mapsto z^r+z$ from $\FF_{r^2}$ to $\FF_r$ is surjective with kernel of
size $r$. There are therefore $r^2-r$ choices for $x_1$. Moreover, if
$x_j^r+x_j\neq0$, then
\[
    \frac{x_j^r}{x_j^{r-1}+1}
    =\frac{x_j^{r+1}}{x_j^r+x_j}\in\FF_r^\times,
\]
so there are exactly $r$ choices for $x_{j+1}$, all with nonzero trace.
These tuples represent distinct affine rational places of $F_i^{(d)}$, and
hence
\[
    |\mathcal{S}_i|=(r^2-r)r^{i-1}=r^i(r-1).
\]

For each proper divisor $b$ of $d$, let $\mathcal{S}_i(b)$ consist of the
tuples in $\mathcal{S}_i$ whose coordinates lie in $\FF_{q^b}$. If
$d/b$ is even, then $z^r=z$ on $\FF_{q^b}$. Thus,
$\mathcal{S}_i(b)$ is empty in characteristic two, whereas in odd
characteristic the recurrence $x_{j+1}=x_j/4$ gives
$|\mathcal{S}_i(b)|=q^b-1$.

If $d/b$ is odd, then $b$ is even. Setting $h=q^{b/2}$, we have
$z^r=z^h$ on $\FF_{q^b}=\FF_{h^2}$. The tower equation is therefore
\[
    x_{j+1}^h+x_{j+1}
    =\frac{x_j^{h+1}}{x_j^h+x_j}.
\]
Here $z^h+z$ and $z^{h+1}$ are the trace and norm from $\FF_{h^2}$ to
$\FF_h$. There are $h^2-h$ choices for $x_1$ of nonzero trace and $h$
choices at every subsequent step. Consequently,
\[
    |\mathcal{S}_i(b)|=h^i(h-1)
    \qquad\text{when $d/b$ is odd}.
\]

A tuple in $\mathcal{S}_i$ has a $q$-Frobenius orbit of size less than $d$
only if it belongs to $\mathcal{S}_i(b)$ for some proper divisor $b$ of
$d$. The even-quotient cases contribute fewer than
\[
    \sum_{b=1}^{d/2}q^b<2r
\]
tuples. In the odd-quotient cases, $b\leq d/3$ and
$h^i(h-1)<q^{b(i+1)/2}$. Hence these cases contribute fewer than
\[
    \sum_{b=1}^{\lfloor d/3\rfloor}q^{b(i+1)/2}
    <2r^{(i+1)/3}
\]
tuples.

For $\boldsymbol{a}=(a_1,\ldots,a_i)\in\mathcal{S}_i$, let
$\widehat P_{\boldsymbol{a}}$ be the affine rational place of $F_i^{(d)}$
represented by $\boldsymbol{a}$, and let
\[
    \operatorname{Orb}_q(\boldsymbol{a})
    =\bigl\{
      (a_1^{q^v},\ldots,a_i^{q^v}):0\leq v<d
     \bigr\}
\]
be its $q$-Frobenius orbit. By
\cref{lem:places-under-constant-extension}, if this orbit has size $d$,
then the places $\widehat P_{\boldsymbol{b}}$ with
$\boldsymbol{b}\in\operatorname{Orb}_q(\boldsymbol{a})$ lie above a
unique degree-$d$ place $P_{\boldsymbol{a}}$ of $F_i$. Moreover, two full
orbits give the same place of $F_i$ if and only if they are equal. Let
$\mathcal{P}_d^{(i)}$ be the set of places obtained from the full orbits in
$\mathcal{S}_i$. We can enumerate this set explicitly by representing
each place by a representative from its orbit. Since every full orbit contains $d$ tuples, for
every $i\geq2$ we have
\[
    d|\mathcal{P}_d^{(i)}|
    >(r-1)r^i-2r-2r^{(i+1)/3}.
\]
The genus formula in \cref{thm:gs-tower-parameters} gives
$g_i/r^i\to1$. Therefore,
\[
    \liminf_{i\to\infty}\frac{|\mathcal{P}_d^{(i)}|}{g_i}
    \geq\frac{r-1}{d}
    =\frac{q^{d/2}-1}{d}.
\]

We now deduce the bound on $B_d(F_i)$. It is trivial for $i=1$, since
$g_1=0$, so assume that $i\geq2$. If $r\geq4$, then
\[
    2r+2r^{(i+1)/3}\leq\frac{r-1}{2}r^i.
\]
The preceding construction of $\mathcal{P}_d^{(i)}$ therefore gives
\[
    dB_d(F_i)
    \geq d|\mathcal{P}_d^{(i)}|
    > (r-1)r^i-2r-2r^{(i+1)/3}
    \geq\frac{r-1}{2}r^i
    \geq\frac{r-1}{2}g_i.
\]
The only cases with $r<4$ are $(q,d)=(2,2)$ and $(q,d)=(3,2)$. In the
first case, $\mathcal{S}_i(1)$ is empty, while in the second it has size
two. Thus,
\[
    2B_2(F_i)\geq
    \begin{cases}
        2^i, & r=2,\\
        2\cdot3^i-2, & r=3,
    \end{cases}
    \geq\frac{r-1}{2}g_i.
\]
Dividing by $d$ proves the general bound.

It remains to prove the assertion about the Riemann--Roch spaces. Let
$\widehat F_i=F_i^{(d)}$, and let $\widehat P_\infty$ be the unique pole
of $x_1$ in $\widehat F_i$. Its restriction $P_\infty$ to $F_i$ is the
unique pole of $x_1$ in $F_i$, and $\widehat P_\infty$ is the unique
place above it. Since $\widehat P_\infty$ is rational,
\cref{lem:places-under-constant-extension} implies that $P_\infty$ is
rational as well. The constant-field-extension theorem for
Riemann--Roch spaces now gives
\begin{equation}
\label{eq:rr-space-constant-extension}
    \mathcal{L}_{\widehat F_i}(\mu\widehat P_\infty)
    =\FF_{q^d} \cdot \mathcal{L}_{F_i}(\mu P_\infty);
\end{equation}
see \cite[Theorem~3.6.3(d)]{Stich}. In particular, the two spaces have
the same dimension over their respective constant fields.

The field $\widehat F_i/\FF_{q^d}$ is a Garcia--Stichtenoth function
field. By \cite{ShumEtAl2001}, we can compute an $\FF_{q^d}$-basis
$f_1,\ldots,f_t$ of
$\mathcal{L}_{\widehat F_i}(\mu\widehat P_\infty)$ in polynomial time.
Let $\sigma$ be the $q$-Frobenius automorphism of $\widehat F_i/F_i$;
that is, $\sigma$ acts on the constants by $c\mapsto c^q$ and fixes the
tower generators $x_1,\ldots,x_i$. By
\cref{eq:rr-space-constant-extension}, the $\sigma$-invariant subspace of
$\mathcal{L}_{\widehat F_i}(\mu\widehat P_\infty)$ is precisely
$\mathcal{L}_{F_i}(\mu P_\infty)$.

Choose an $\FF_q$-basis $\beta_1,\ldots,\beta_d$ of $\FF_{q^d}$. For
$a\in[d]$ and $j\in[t]$, compute
\[
    h_{a,j}
    =\Tr_{\widehat F_i/F_i}(\beta_af_j)
    =\sum_{v=0}^{d-1}\sigma^v(\beta_af_j).
\]
Each $h_{a,j}$ belongs to $\mathcal{L}_{F_i}(\mu P_\infty)$. Moreover,
these elements span this space over $\FF_q$. Indeed, the elements
$\beta_af_j$ form an $\FF_q$-basis of
$\mathcal{L}_{\widehat F_i}(\mu\widehat P_\infty)$, so their traces span
the image of $\Tr_{\widehat F_i/F_i}$. This image is the entire descended
Riemann--Roch space: for any $f\in\mathcal{L}_{F_i}(\mu P_\infty)$, choose
$\beta\in\FF_{q^d}$ with $\Tr_{\FF_{q^d}/\FF_q}(\beta)=1$, which exists
because finite fields are separable. Then
\[
    \Tr_{\widehat F_i/F_i}(\beta f)=f.
\]
We can therefore select an $\FF_q$-basis from the elements $h_{a,j}$ by
Gaussian elimination.

Finally, let $P$ be a degree-$d$ place represented by a tuple
$\boldsymbol{a}\in\FF_{q^d}^i$, and let $\widehat P_{\boldsymbol{a}}$ be
the corresponding rational place of $\widehat F_i$. Under the
identification $\kappa(P)\cong\FF_{q^d}$ determined by
$\boldsymbol{a}$, evaluation at $P$ agrees with evaluation at
$\widehat P_{\boldsymbol{a}}$. The latter can be computed efficiently by
the algorithm of \cite{ShumEtAl2001}.
\end{proof}

\end{document}